\documentclass[11pt,a4paper]{article}
\pdfoutput=1 %for arxiv
\usepackage[utf8]{inputenc}

\usepackage[margin=3cm]{geometry}
\usepackage{amsfonts}
\usepackage{amsmath}
\usepackage{amssymb}
\usepackage{amsthm}
\usepackage{float}
\usepackage[font=small,labelfont=bf]{caption} %customise text style on fig captions
\allowdisplaybreaks

\usepackage{physics}
\usepackage{xspace}

\usepackage{thmtools}
\usepackage{thm-restate}

\usepackage{bold-extra} %adds bold+smallcaps (for standard fonts)

\usepackage{hyperref}
\usepackage[svgnames]{xcolor}
\hypersetup{colorlinks={true},urlcolor={blue},linkcolor={DarkBlue},citecolor=[named]{DarkGreen}}
\usepackage[authoryear,square]{natbib}

\usepackage{microtype}
\usepackage[capitalise,nameinlink,noabbrev]{cleveref}

\usepackage{tikz}
\usetikzlibrary{arrows.meta}

\usepackage{doi}

\theoremstyle{definition}
\newtheorem{definition}{Definition}

\theoremstyle{plain}
\newtheorem{theorem}{Theorem}[section]
\newtheorem{lemma}[theorem]{Lemma}

\newtheorem{claim}{Claim}

\theoremstyle{remark}
\newtheorem{remark}{Remark}

\Crefname{claim}{Claim}{Claims}

\newcommand{\sz}{\textup{size}}

\def\fp/{\textup{\textsf{FP}}}
\def\p/{\textup{\textsf{P}}}
\def\np/{\textup{\textsf{NP}}}
\def\conp/{\textup{\textsf{co-NP}}}
\def\fnp/{\textup{\textsf{FNP}}}
\def\tfnp/{\textup{\textsf{TFNP}}}
\def\ptfnp/{\textup{\textsf{PTFNP}}}
\def\ppa/{\textup{\textsf{PPA}}}
\def\ppad/{\textup{\textsf{PPAD}}}
\def\ppads/{\textup{\textsf{PPADS}}}
\def\ppp/{\textup{\textsf{PPP}}}
\def\pwpp/{\textup{\textsf{PWPP}}}
\def\pls/{\textup{\textsf{PLS}}}
\def\cls/{\textup{\textsf{CLS}}}
\def\ppadpls/{\textup{$\textsf{PPAD} \cap \textsf{PLS}$}}
\def\ppapls/{\textup{$\textsf{PPA} \cap \textsf{PLS}$}}
\def\eopl/{\textup{\textsf{EOPL}}}
\def\sopl/{\textup{\textsf{SOPL}}}
\def\ueopl/{\textup{\textsf{UEOPL}}}
\def\fixp/{\textup{\textsf{FIXP}}}
\def\bu/{\textup{\textsf{BU}}}
\def\bbu/{\textup{\textsf{BBU}}}
\def\linearfixp/{\textup{\textsf{Linear-FIXP}}}
\def\pspace/{\textup{\textsf{PSPACE}}}

\def\gcircuit/{\textup{\textsc{Gcircuit}}}

\providecommand{\ceil}[1]{\ensuremath{\left \lceil #1 \right \rceil }}

\def\ritucker{\ensuremath{\textup{\textsc{StrongTucker}}}\xspace}
\def\ndritucker{\ensuremath{\textup{\textsc{$N$D-StrongTucker}}}\xspace}
\def\twodritucker{\ensuremath{\textup{\textsc{$2$D-StrongTucker}}}\xspace}

\def\tucker{\ensuremath{\textup{\textsc{Tucker}}}\xspace}
\def\ndtucker{\ensuremath{\textup{\textsc{$N$D-Tucker}}}\xspace}
\def\twodtucker{\ensuremath{\textup{\textsc{$2$D-Tucker}}}\xspace}
\def\conhalv{\ensuremath{\textup{\textsc{Consensus-Halving}}}\xspace}

\def\necklace{\ensuremath{\textup{\textsc{NecklaceSplitting}}}\xspace}

\def\fsp{\ensuremath{\textup{\textsc{FairSplitPath}}}\xspace}
\def\fsc{\ensuremath{\textup{\textsc{FairSplitCycle}}}\xspace}
\def\straightpizza{\ensuremath{\textup{\textsc{StraightPizzaSharing}}}\xspace}
\def\squarepizza{\ensuremath{\textup{\textsc{SquarePizzaSharing}}}\xspace}
\def\hamsand{\ensuremath{\textup{\textsc{DiscreteHamSandwich}}}\xspace}

\newcommand{\CHinstance}{\ensuremath{\textup{CH}_\eps(\lambda)}\xspace}
\newcommand{\CHinstancenoeps}{\ensuremath{\textup{CH}(\lambda)}\xspace}
\newcommand{\mlambda}{\ensuremath{\widehat{\lambda}}}
\newcommand{\val}{\ensuremath{\textup{\textsf{val}}}}

\newcommand{\wt}{\ensuremath{8}\xspace}

\newcommand{\eps}{\ensuremath{\varepsilon}\xspace}

\newcommand{\reals}{\ensuremath{\mathbb{R}}\xspace}
\newcommand{\naturals}{\ensuremath{\mathbb{N}}\xspace}

\newcommand{\rplus}{\ensuremath{R^+}\xspace}
\newcommand{\rminus}{\ensuremath{R^-}\xspace}

\title{Constant Inapproximability for \textsf{PPA}\thanks{A preliminary version of this paper appeared at STOC '22.}}

\author{
\begin{tabular}{cc}
& \\
\textbf{Argyrios Deligkas} & \textbf{John Fearnley}\\
\small{Royal Holloway, United Kingdom} & \small{University of Liverpool, United Kingdom}\\
\href{mailto:argyrios.deligkas@rhul.ac.uk}{\small{\texttt{argyrios.deligkas@rhul.ac.uk}}} & \href{mailto:john.fearnley@liverpool.ac.uk}{\small{\texttt{john.fearnley@liverpool.ac.uk}}}\\
& \\
\textbf{Alexandros Hollender} & \textbf{Themistoklis Melissourgos}\\
\small{University of Oxford, United Kingdom} & \small{University of Essex, United Kingdom}\\
\href{mailto:alexandros.hollender@cs.ox.ac.uk}{\small{\texttt{alexandros.hollender@cs.ox.ac.uk}}} & \href{mailto:themistoklis.melissourgos@essex.ac.uk}{\small{\texttt{themistoklis.melissourgos@essex.ac.uk}}}\\
& \\
\end{tabular}
}

\date{}

\begin{document}

\maketitle

\begin{abstract}
In the $\varepsilon$-Consensus-Halving problem, we are given $n$ probability measures $v_1, \dots, v_n$ on the interval $R = [0,1]$, and the goal is to partition $R$ into two parts $R^+$ and $R^-$ using at most $n$ cuts, so that $|v_i(R^+) - v_i(R^-)| \leq \varepsilon$ for all $i$. This fundamental fair division problem was the first natural problem shown to be complete for the class \ppa/, and all subsequent \ppa/-completeness results for other natural problems have been obtained by reducing from it.

We show that $\varepsilon$-Consensus-Halving is \ppa/-complete even when the parameter $\varepsilon$ is a constant. In fact, we prove that this holds for any constant $\varepsilon < 1/5$. As a result, we obtain constant inapproximability results for all known natural \ppa/-complete problems, including Necklace-Splitting, the Discrete-Ham-Sandwich problem, two variants of the pizza sharing problem, and for finding fair independent sets in cycles and paths.
\end{abstract}

\newpage

\section{Introduction}

The consensus halving problem \citep{SS03-Consensus} is a fair division problem defined by $n$ agents,
who each have a valuation function over the unit interval $R = [0, 1]$. The goal is to
partition $R$ into two sets $R^+$ and $R^-$  using at most $n$ cuts, such
that all agents agree that $R^+$ and $R^-$ have the same valuation, or in the
$\eps$-approximate version, that all agents agree that $R^+$ and $R^-$ have
valuations that differ by at most $\eps$.

The problem is guaranteed to have a solution and this is usually proved by using the Borsuk-Ulam theorem from topology, or its discrete counterpart, Tucker's lemma \citep{SS03-Consensus}. In fact, very similar versions of this existence result have been proved in the past in different contexts \citep{hobby1965moment,alon1986borsuk,Alon87-necklace}. Since the problem is guaranteed to have a solution and solutions can be verified efficiently, it lies in the complexity class \tfnp/: the class of total \np/ search problems. In particular, this means that the problem cannot be \np/-hard, unless $\np/ = \conp/$ \citep{MegiddoP91-TFNP}, and instead, one has to use subclasses of \tfnp/ to classify its complexity.

The consensus halving problem has risen to prominence as it has played a crucial
role in the development of the complexity class \ppa/, a subclass of \tfnp/ defined by~\citet{Papadimitriou94-TFNP-subclasses}. Indeed, in a breakthrough result, \citet{FRG18-Consensus} proved that the problem is complete for \ppa/. This was the first ``natural'' complete problem for the class and it has been pivotal in proving further such completeness results. For
example, \ppa/-completeness has since been shown for other ``natural'' problems such as
the necklace splitting problem and the discrete ham sandwich
problem~\citep{FRG19-Necklace}, two types of the pizza-sharing
problem~\citep{DeligkasFM22-pizza,schnider2021complexity}, and finding fair independent sets
in cycles and paths~\citep{haviv2020complexity}. We refer to these as natural problems\footnote{We note that some of these problems also have more general ``unnatural'' versions where the inputs, e.g., the valuations, are represented by circuits. The aforementioned completeness results apply to the natural versions without circuits.}
since their definition does not involve any kind of circuit, as opposed to
``unnatural'' problems like \tucker (the problem associated with Tucker's Lemma), which was already known to be
\ppa/-complete~\citep{AisenbergBB20-2D-Tucker}, but whose definition involves a
Boolean circuit. 

Consensus halving has been used in a fundamental way to show \ppa/-completeness
for natural problems, because it bridges the gap between natural and unnatural
\ppa/-complete problems. Specifically, the \ppa/-hardness results for consensus
halving~\citep{FRG18-Consensus,FRG19-Necklace} reduce from \tucker, and
explicitly remove the Boolean circuit by encoding each gate as a consensus
halving agent. To the best of our knowledge, \emph{all} subsequent
hardness results for natural problems have reduced from consensus halving.

\paragraph{\bf Hardness of approximation.}

Prior work has shown that, not only is it \ppa/-complete to find exact consensus
halving solutions for piecewise constant valuation functions, but it is also \ppa/-complete to find approximate solutions.
The initial hardness result of~\citet{FRG18-Consensus} showed that
$\eps$-\conhalv is \ppa/-complete for an exponentially small~\eps (in the size of the input). The
same authors later improved this to obtain a \ppa/-completeness result for
$\eps$-\conhalv with~$\eps$ being polynomially
small~\citep{FRG19-Necklace}.

The hardness of approximation for consensus halving has then directly led to
hardness of approximation for the other natural \ppa/-complete problems, because
all of the \ppa/-hardness reductions for natural problems that have been discovered so far preserve approximate solutions.
So, we have that 
necklace splitting, discrete ham sandwich,
pizza-sharing,
and finding fair independent sets in cycles and paths are all \ppa/-complete to approximate
for a polynomially small $\eps$.  

In this sense, consensus halving plays a crucial role in the hardness of
approximation for natural \ppa/-complete problems, because any improvement in
the hardness result for consensus halving directly leads to an improvement in
the hardness results for \emph{all} of the natural problems that are currently
known to be \ppa/-complete.

The key question left open by previous work is whether 
$\eps$-\conhalv is \ppa/-hard, and thus \ppa/-complete, for a constant $\eps$. While there is no
such result in prior work, 
consensus halving is known to be \ppad/-hard for a very small 
constant~$\eps$~\citep{FRFGZ18-consensus-hardness}. 
This result actually predates all of the \ppa/-hardness results and arises from a
direct reduction to $\eps$-\conhalv from the \gcircuit/ problem, which
is known to be \ppad/-complete for constant
$\eps$~\citep{Rubinstein18-Nash-inapproximability}. Notably, though, the constant
is so small that no prior work has actually given a lower bound on its
magnitude.

Furthermore, even ignoring the minuscule $\eps$, this result is somewhat unsatisfying, since it seems unlikely that
\ppad/-hardness is the correct answer for constant approximation, given that $\ppad/
\subseteq \ppa/$, and \ppa/ appears to capture a strictly larger class of
problems. 
This is doubly so, since
finding a polynomially small approximation is known to be \ppa/-complete, and
thus \ppa/-completeness of finding constant approximations would be the natural, and tight, answer.

\paragraph{\bf Our Contribution.} 

Our main result is as follows.

\begin{theorem}\label{thm:main-result}
$\eps$-\conhalv is \ppa/-complete for all $\eps < 1/5$.
\end{theorem}

Thus, we show hardness for a constant $\eps$, improving upon the prior
state-of-the-art result, which showed hardness for a polynomially small $\eps$.
A direct consequence of this theorem is that the hardness results for \emph{all}
natural problems that are known to be \ppa/-complete are strengthened as well,
and we obtain \ppa/-hardness for each of the problems for a constant $\eps$.

Our result shows hardness for any $\eps < 1/5$, which is notably large compared
to other constant inapproximability results for total search problems. For
example, the current state-of-the-art hardness results for \ppad/-complete
problems do show hardness for a
constant~$\eps$~\citep{Rubinstein18-Nash-inapproximability}, but as mentioned earlier, that constant
is so small that no prior work has given a lower bound on its
magnitude.\footnote{In subsequent work, we obtained similar strong explicit inapproximability bounds for \ppad/ problems \citep{DeligkasFHM24-pure-circuit}.} Here we give a constant that is substantial relative to the trivial
upper bound of $\eps=1$. We obtain similarly large constants for each of the
other natural problems that are known to be \ppa/-complete, as shown in the
following table.
\begin{center}
\begin{tabular}{l|l}
& \ppa/-completeness \\
Problem & threshold \\\hline
\eps-\necklace & 1/5 \\
\eps-\hamsand & 1/5 \\
\eps-\straightpizza & 1/5 \\
\eps-\squarepizza & 1/5 \\
\eps-\fsp & 1/20 \\
\eps-\fsc & 1/20
\end{tabular}
\end{center}
The full details of these follow-on hardness results can be found below in
\cref{sec:consequences}.

Moreover, our main result continues to hold even if we severely restrict the valuation functions of the agents.

\begin{theorem}\label{thm:main-result-3-block}
$\eps$-\conhalv is \ppa/-complete for all $\eps < 1/5$, even if all agents have 3-block uniform valuations.
\end{theorem}

An agent has a 3-block uniform valuation function if the density function of the valuation is non-zero in at most three intervals, and in each such interval it has the same non-zero value.

Finally, by a standard argument \citep{FRHSZ20-consensus-easier}, it immediately follows that the hardness result holds even if we allow a few more than just $n$ cuts.

\begin{restatable}{theorem}{extracuts} \label{thm:extra_cuts}
$\eps$-\conhalv is \ppa/-complete for all $\eps < 1/5$, even if all agents have 3-block uniform valuations, and even if $n+n^{1-\delta}$ cuts are allowed for some constant $\delta \in (0,1]$, where $n$ is the number of agents.
\end{restatable}

For completeness we provide a proof of this result in \cref{sec:app:extra-cuts}.

\subsection{Direct Consequences} 
\label{sec:consequences}

Our hardness result for $\eps$-\conhalv directly yields improved hardness
results for every natural problem that is currently known to be \ppa/-complete.
In this section we give the details for these improved hardness results.

\paragraph{\bf Necklace Splitting.} 

In \necklace, we are given a necklace with beads of $n$ colours, and we want to
split the necklace into two (in general, non-contiguous) parts by making at most
$n$ cuts, such that both parts contain half of the beads of each colour. It was
shown by \citet{goldberg1985bisection} and \citet{alon1986borsuk} that the problem
always admits a solution, and later \citet{Alon87-necklace} extended this result
to the variant where the necklace must be divided into $k$ parts rather than
two. 

\ppa/-completeness for the problem was proven by \citet{FRG19-Necklace} via a
reduction from \eps-\conhalv~for an inversely-polynomial \eps. In addition, in
\citep{FRG18-Consensus} it was proven that the approximate version of the
problem is \ppad/-hard for some small constant \eps.

In the approximate version of the problem, denoted as \eps-\necklace with $\eps
\in (0,1)$, the goal is to cut the necklace into two
parts such that, for each colour, the discrepancy between the two
parts is bounded by $\eps$. Formally, if there are $B_i$ beads of colour $i$ and $B_i^+, B_i^-$
correspond to the number of beads of colour $i$ in each of the two parts, in an \eps-solution it
holds that $|B_i^+ - B_i^-| \leq \eps \cdot B_i$. The reduction presented
in~\citep{FRG18-Consensus} increases the error of the \eps-\conhalv instance by only a polynomially small
amount\footnote{We note
that~\citep{FRG18-Consensus} have defined \eps-\necklace
with $\eps$ denoting the discrepancy between the number of beads in each of the two parts, rather than normalising $\eps$ so that it is expressed relative to the total number of beads.  However, this appears to have been a mistake, since their proof and result actually 
use the
definition that we give
here.}, so by applying our our main result, we obtain the following.

\begin{theorem}
\eps-\necklace is \ppa/-complete for every constant $\eps < 1/5$, even if $n+n^{1-\delta}$ cuts are allowed for some constant $\delta > 0$.
\end{theorem}

\paragraph{\bf Ham Sandwich.} 
In \hamsand, as defined by~\citet{Papadimitriou94-TFNP-subclasses}, we are given
$n$ sets of points with integer coordinates in $d$ dimensional space, where $d
\geq n$.
The task is to find a hyperplane that cuts the space into two 
halfspaces, such that each halfspace contains half of the points of each set. If any
points lie on the plane, then we are allowed to place each of them on either side. 
\citep{FRG19-Necklace} proved that the problem is \ppa/-complete, via a
reduction from \necklace. 

In the approximate version of the problem, denoted
\eps-\hamsand, we want to find a hyperplane such that, for every set, the
discrepancy between the number of points contained in the two halfspaces is bounded by
\eps. Formally, if there are $S_i$ points for set $i$ and $S_i^+, S_i^-$
correspond to the number of points belonging to the two halfspaces, in an
\eps-solution we must have $|S_i^+ - S_i^-| \leq \eps \cdot S_i$. The reduction
between \hamsand and \necklace presented by \citet{FRG19-Necklace} is 
approximation preserving, so we get the following theorem.

\begin{theorem}
\eps-\hamsand is \ppa/-complete for every constant $\eps < 1/5$, even if $d = n+n^{1-\delta}$ for some constant $\delta > 0$.
\end{theorem}

\paragraph{\bf Pizza Sharing.} 

In pizza sharing problems we are given measurable objects that are embedded in
the two-dimensional plane, and we are asked to make a number of cuts in order to
divide each mass into two equally sized portions. Two versions of this problem
have been studied in the literature. 

In the \straightpizza problem, we are given $2n$ two-dimensional masses in the
plane, and we are asked to find $\ell$ straight lines that create a
``checkerboard'' that  simultaneously bisects all of the masses. In~\citep{BPS19,
HK20} it was shown that the problem always admits a solution when $\ell \geq n$.

In the \squarepizza problem, there are $n$ masses in the plane, and the task is
to simultaneously bisect all masses via a \emph{square-cut}: a path that is the
union of horizontal and vertical line segments. 
In~\citep{KRS16} it was
proven that a path with $n-1$ turns can always bisect all $n$ masses. 

In the approximate versions of these problems, we are looking for an approximate
bisection. Formally, in an $\eps$-approximate solution of these problems, we are
looking for a partition of the plane $R$ into two regions $\rplus$ and $\rminus$
such that for every measure $\mu_i$ it holds that $|\mu_i(\rplus) -
\mu_i(\rminus)| \leq \eps \cdot \mu_i(R)$. For \eps-\straightpizza the partition
must be produced by $\ell$ lines, while for \eps-\squarepizza the partition must
be produced by a square-cut path with $t$ turns. 

Both problems were proven to be
\ppa/-complete when \eps is inversely polynomial and \ppad/-hard for a small
constant $\eps \in (0,1)$ via direct reductions from
\conhalv~\citep{DeligkasFM22-pizza,schnider2021complexity}. Using the reductions
from~\citep{DeligkasFM22-pizza}, which increase the error by at most a polynomially small amount, alongside our main theorem yields the following.

\begin{theorem}
\label{thm:straightpizza}
\eps-\straightpizza is \ppa/-complete for every constant $\eps < 1/5$, even if $n+n^{1-\delta}$ cuts are allowed for some constant $\delta > 0$.
\end{theorem}

\begin{theorem}
\label{thm:squarepizza}
\eps-\squarepizza is \ppa/-complete for every constant $\eps < 1/5$, even if the square-cut path is allowed to have $n+n^{1-\delta}$ turns for some constant $\delta > 0$.
\end{theorem}

\paragraph{\bf Fair Independent Sets.}
In this setting, we are given a graph $G$ whose vertices are partitioned into
$n$ sets $V_1, \ldots, V_n$. The task is to find two independent sets of $G$
such that every $V_i$ is covered in a ``fair'' manner. In particular, we are
interested in the setting where $G$ is a
cycle or a path, since it was proven that such graphs
possess fair independent
sets~\citep{aharoni2017fair,alishahi2017fair,black2020fair}.

More specifically,~\citet{alishahi2017fair} proved that if $G$ is a cycle
of $m$ vertices and $n$ has the same parity as $m$, then there exist two
disjoint independent sets $S_1$ and $S_2$, such that for every $i \in [n]$ it holds that $|V_i \cap (S_1 \cup S_2)|
= |V_i|-1$ and $|S_j \cap V_i| \geq \frac{1}{2} |V_i| - 1$ for all $j \in \{1,2\}$. For $\eps \in
[0,\frac{1}{2}]$, we use \eps-\fsc to denote the problem of finding two such
independent sets, where the second condition is relaxed to $|S_j \cap V_i| \geq (\frac{1}{2} - \eps)\cdot |V_i| - 1$.

A similar theorem was shown for paths by \citet{black2020fair}:
if $G$ is a path and every set $V_i$ contains an odd number of points, then there exist two independent sets
$S_1$ and $S_2$, covering all but at most $n$ vertices of $G$ such that for every $i \in [n]$ it holds that $|S_1 \cap
V_i| \in \big[\frac{1}{2}|V_i| - 1,  \frac{1}{2}
|V_i| \big]$. For $\eps \in [0,\frac{1}{2}]$, we use \eps-\fsp to denote the corresponding computational
problem, where the condition is relaxed to $|S_1 \cap
V_i| \in \big[(\frac{1}{2} - \eps)\cdot |V_i| - 1,  (\frac{1}{2} + \eps)\cdot
|V_i| \big]$.

Hardness was shown for both problems by \citet{haviv2020complexity}, who proved that both problems are \ppa/-complete
for a polynomially small \eps and that they are \ppad/-hard for a small
constant \eps. The hardness is shown by a reduction from
\eps-\conhalv to $\frac{\eps}{4}$-\fsp and then a follow-on reduction from \eps-\fsp to \eps-\fsc.
Combining these reductions with our main theorem yields the following.\footnote{In a subsequent version of his work, \citet{Haviv2022-fair-cycles} improved the parameters of his reduction from \conhalv to \fsp, and thus obtained \ppa/-hardness results for $\eps < 1/10$.}

\begin{theorem}
\label{thm:fairis}
\eps-\fsp and \eps-\fsc are \ppa/-complete for every constant $\eps < 1/20$.
\end{theorem}

\subsection{Further Related Work}

There are various other works that relate to ours, some of which dealt with \ppa/-hardness and some of which studied the consensus halving problem. 

In a recent work by \citet{DeligkasFH22-consensus-constant} the main question was ``How does the complexity of consensus halving depend on the number of agents?''. This paper's main result is a dichotomy between 2 and 3 agents when the valuations are monotone (but possibly non-additive). In particular, for the former case the problem is polynomial time solvable, while for the latter it is \ppa/-complete. If the monotonicity property is dropped, then both cases become \ppa/-complete. Furthermore, for the case of a single agent (and even for $n$ agents with identical valuations) the problem is polynomial time solvable. 

\citet{AlonG21} present a set of strong positive results on the \eps-\hspace{0pt}\textsc{NecklaceSplitting} problem. They present efficient algorithms for a relaxed version of this problem where more than $n$ cuts are allowed. In particular, for an instance whose beads can take $n$ colours, and can be at most $m$ per colour, they give an offline and an online algorithm that is efficient and deterministic, which provide a solution by making at most $O(n (\log m + O(1)))$ and $O(m^{2/3} \cdot n(\log n)^{1/3})$ cuts, respectively, for $\eps = 0$. For $\eps > 0$, the same algorithms work with the aforementioned number of cuts, by substituting $m$ with $1/\eps$. These algorithms also work for the \eps-\conhalv problem when we are allowed to use more than $n$ cuts. Their positive results extend to the generalization of \necklace in which, instead of wishing to split each colour's beads into two parts, we split them into $k \geq 2$ parts \citep{Alon87-necklace}. For detailed definitions of this and related problems, as well as their related complexity classes, see \citep{FRHSZ21-necklace}, and \citep{Hollender21-ppa-k}.

In \citep{GoldbergHIMS22-consensus-items} the problem under study deviates slightly from the typical consensus halving problem. There are (divisible) items and they are not presented in a linear order, but rather unordered, with agents having linear and additively separable utilities over them. In this work the authors provide polynomial time algorithms even for the more general \conhalv problem where we do not split the probability measures in two, but in $k \geq 2$ parts, and show that for a slightly non-linear valuation class the problem becomes \ppad/-hard. For the case where the items are in a specific order, they show that the problem is \ppa/-complete.

There is also a hierarchy of complexity classes of problems that seek exact solutions, whose output involves irrational numbers. A famous such class is \fixp/ whose typical problem has as input a function from $[0,1]^N$ to itself, and the task is to find an exact Brouwer fixed point of the function. This class was defined by \citet{EtessamiY10-FIXP} who further showed that the problem of finding an exact Nash equilibrium for $n \geq 3$ agents is complete for the class. In \citep{deligkas2021BU} a related class was defined, namely \bu/, whose typical problem has as input a function from the $L_1$ unit $N$-sphere to $\reals^N$, and the task is to find an exact Borsuk-Ulam point, i.e., one that has the same function value as its antipodal. This work showed that exact \conhalv is in \bu/ for piecewise polynomial valuations and \fixp/-hard. Both the aforementioned papers showed that when the input function is piecewise linear, then the induced class is identical to \ppad/ and \ppa/ respectively. \citet{EtessamiY10-FIXP} also defined the strong approximation version of an exact search problem where a point that is close to an exact solution is sought. By extending this notion to \bu/, \citet{BatziouHH21-consensus-BBU} showed that the strong approximation version of \conhalv (with valuations represented by algebraic circuits) is complete for the corresponding class $\bu/_{a}$.

\section{Preliminaries}\label{sec:prelims}

A {\em valuation function}, or simply valuation, of an agent is a probability
measure over the interval $R=[0,1]$. The probability measures are given by their
density functions. A valuation function is {\em piecewise constant} if $R$ can
be partitioned into a finite set of intervals where the density of the probability is
constant over each interval. Thus, a piecewise constant valuation can be
explicitly represented as endpoints and heights of value blocks. For any measurable subset $S$ of $R$, $v_i(S)$ denotes the value of agent $i$ for $S$;
$v_i(S)$ equals the measure of the probability of agent $i$ over $S$. In particular, $v_i(R) = 1$. For $n \in \mathbb{N}$, we use $[n]$ to denote the set $\{1,2,\dots,n\}$.

\begin{definition}[\eps-\conhalv]
An instance of \eps-\conhalv consists of $n$ agents with piecewise constant
valuation functions over the interval $R=[0,1]$. A solution 
is a partition of the interval $R$ into two regions $R^+$ and $R^-$, using at most $n$
cuts, where every agent agrees that 
the value of $R^+$ is at most $\eps$-away from the value for $R^-$. Formally, in a solution
of \eps-\conhalv it holds that $|v_i(R^+)-v_i(R^-)| \leq \eps$ for every $i \in
[n]$.
\end{definition}

In this paper we will show a hardness result for \eps-\conhalv by reducing from
the \twodtucker problem.

\begin{definition}[\twodtucker]
An instance of \twodtucker consists of a labelling function $\lambda : [m] \times
[m] \to \{\pm 1, \pm 2\}$ such that for $1 \leq i,j \leq m$, $\lambda(i,1) =
-\lambda(m-i+1,m)$ and $\lambda(1,j) = -\lambda(m,m-j+1)$. A solution to such an
instance is a pair of vertices $(x_1,y_1)$, $(x_2,y_2)$ with $|x_1 - x_2| \leq
1$ and $|y_1 - y_2| \leq 1$ such that $\lambda(x_1,y_1) = -\lambda(x_2,y_2)$.
\end{definition}

The labelling $\lambda$ is given as a Boolean circuit.
\twodtucker is known to be \ppa/-complete;
\citet{Papadimitriou94-TFNP-subclasses} proved membership in \ppa/ and
\citet{AisenbergBB20-2D-Tucker} proved \ppa/-hardness.

\begin{theorem}[\citet{AisenbergBB20-2D-Tucker,Papadimitriou94-TFNP-subclasses}]
\twodtucker is \ppa/-complete.
\end{theorem}

Other versions of Tucker's lemma have also been shown to be \ppa/-complete \citep{dengoctahedral,FRG19-Necklace}.

\section{Technical Overview}

In this section, we present an overview of the proof of our main result including the new insights that allow us to obtain hardness for a constant $\eps$.

To prove our main result, we reduce \twodtucker
to $\eps$-\conhalv for
all $\eps < 1/5$. Here we give an overview of the reduction, and the key
challenges that needed to be overcome in order to obtain a constant $\eps$. 

In our description of the main ideas and challenges, we will make reference to the three existing \ppa/-hardness reductions for consensus halving.\footnote{The paper~\citep{FRG22-NS-CH-ham} is a journal version, which combines the proof from Work 2 (which reuses some machinery from Work 1) with some results from Work 1, such as the connection between \conhalv and \necklace.}
\begin{itemize}
    \item \textbf{Work 1} \citep{FRG18-Consensus}: which proves hardness for inverse exponential $\eps$.
    \item \textbf{Work 2} \citep{FRG19-Necklace}: which proves hardness for inverse polynomial $\eps$.
    \item \textbf{Work 3} \citep{FRHSZ20-consensus-easier}: which provides a significantly simplified proof of hardness for inverse polynomial $\eps$.
\end{itemize}

All three existing works ultimately reduce from \twodtucker, but Works 2 and 3 include a preliminary step, where \twodtucker is reduced to its high-dimensional version: \ndtucker. This seems to be necessary in order to obtain hardness for inverse polynomial \eps. Indeed, a similar observation can also be made about analogous results in the study of approximate Nash equilibrium computation \citep{DaskalakisGP09-Nash,ChenDT09-Nash}, where a high-dimensional version of the Brouwer problem is used to achieve hardness for inverse polynomial approximation.

We begin with a very high-level overview of the general structure of the reduction which applies to all three existing works, as well as to ours. Informally, an \ndtucker
instance is defined over an $N$-dimensional grid $G = [m] \times [m] \times \dots
\times [m]$ with side length $m$. The instance gives a \emph{labelling function}
$\lambda : G \rightarrow \{-N, \dots, -1, 1, \dots, N\}$, presented as a Boolean
circuit, that assigns each point in the grid a label that is either $+i$ or $-i$
for some $i$ in the range $1 \le i \le N$. Additionally, the labelling satisfies an \emph{antipodality} condition on the boundary: letting $\overline{x_i} := m - x_i + 1$, it holds that $\lambda(\overline{x}) = - \lambda(x)$ whenever $x$ lies on the boundary of $G$.
The goal is to find two points $x$
and $y$ on the grid, such that $x$ and $y$ are within $L_\infty$ distance 1 of
each other, and $\lambda(x) = -\lambda(y)$. Such a pair of points is guaranteed to exist by Tucker's Lemma~\citep{tucker1945some}, and the problem of finding one is \ppa/-complete even for constant $m$, as shown in Work 2 by reducing from \twodtucker. The problem \twodtucker is defined in the same way, except that $N=2$, and $m$ is required to be exponentially large for the problem to be \ppa/-complete \citep{AisenbergBB20-2D-Tucker}.

We are now ready to present the high-level setup used in all three previous works. The specifics of the reductions in Works 1 and 2 are significantly more involved than what is presented here, and so the presentation below should be seen as mostly applying to the simplified proof of Work 3 (while still representing the underlying core structure hidden behind the reductions in Works 1 and 2). The $\eps$-\conhalv instance $\CHinstancenoeps$ is constructed as follows:
\begin{itemize}
\item The line $R = [0,1]$ consists of two intervals $I$ and $C$. We think of $I$ as the \emph{input region} (also called \emph{coordinate-encoding region} in prior work), and $C$ as the \emph{circuit region} (also called \emph{circuit-encoding region}). In any solution $S = (R^+,R^-)$ to the instance $\CHinstancenoeps$, $I$ will be partitioned into two regions $I^+ := I \cap R^+$ and $I^- := I \cap R^-$. The exact way in which $I$ is partitioned encodes a point $z := z(I^+,I^-)$ in some domain. In Work 1, this domain is a locally two-dimensional Möbius strip, while in Work 2, it is a high-dimensional generalization of that. Work 3 significantly simplifies this encoding by letting the domain simply be the $N$-dimensional unit hypercube. In all three cases, the grid $G$ of the \ndtucker instance $\lambda$ is embedded in the domain in question, and so the partition $(I^+, I^-)$ ultimately encodes a point $x := x(I^+,I^-) \in G$.

\item To ``extract'' the point $x \in G$ from $(I^+,I^-)$, a binary
decoding step is performed where the continuous information that is encoded in $(I^+,I^-)$ is converted into bit values that represent $x_i$. 
Mechanically, this is implemented by introducing a set of agents in $\CHinstancenoeps$ who ensure that cuts (between $R^+$ and $R^-$) are
placed in specified regions of $C \subset [0, 1]$ to encode either a $0$ bit or a $1$ bit. In Works 1 and 2, this binary decoding is performed ``at the source'', namely the information read from $I$ is essentially binary, and then further processed by simulating Boolean gates. In Work $3$, the information read from $I$ is continuous and then further processed by simulating arithmetic gates. The arithmetic gates are then used to perform the bit decoding.

\item As is common in \ppad/ and \ppa/ reductions, the Boolean decoding step from $z$ to $x$ can
fail for certain values, and if the decoding step fails then nonsensical values
will be produced. To address this, instead of just decoding $z$, a large number of points
surrounding the encoded point $z$ are decoded, where the samples are chosen to
ensure that only a small number of the decoding steps fail. If $K$ samples are
taken, then this gives a sequence of points $x^1, x^2, \dots, x^K$, most of
which are valid bit representations of points in $G$, and a small number of which contain nonsensical values. Importantly, the sampling is performed such that all resulting (correctly) decoded points in $G$ lie within $L_\infty$ distance $1$ of each other.

\item The next step is to simulate the execution of the Tucker labelling
circuit $\lambda$ on the inputs $x^1, x^2, \dots, x^K$. For this, $K$ completely independent copies of the circuit $\lambda$ are simulated, each within its own sub-region $C^1, \dots, C^K$ of the circuit region $C$. The region $C^k$ is fed input $x^k$ and so is used to compute
$\lambda(x^k)$. Mechanically, each gate in each circuit is simulated by a set of
agents who ensure that the output of each gate is encoded by a cut in a
specified region of $C^k$, allowing other agents to read that value to
simulate other gates.

\item The output of the circuit in $C^k$ is treated as a vector $y^k \in [-1,1]^N$, so
that whenever $x^k$ was correctly decoded and $\lambda(x^k) = +i$ (resp. $-i$), we have $y^k_i = +1$ (resp. $-1$) and $y^k_\ell = 0$ for all other dimensions $\ell \neq i$. If $x^k$ was
not correctly decoded, then $y^k$ can be any vector in $[-1,1]^N$.

\item The last step is to average the outputs of the circuits. Specifically, for each dimension $i \in [N]$, we introduce an agent that computes the average $L(i) = \frac{1}{K} \sum_{k = 1}^K y^k_i$ and enforces that $L(i)$ be $\eps$-close to zero for all $i$. With $K$ being chosen to be
suitably large, the effect of the incorrectly decoded points becomes negligible,
and so 
it is only possible for $L(i)$ to be close to zero for all $i$ if, for each
dimension $i$, there are a roughly equal number of points with label $+i$ and label
$-i$. Since all of the input points lie within $L_\infty$ distance $1$ of each
other, this implies that we have a solution to the Tucker instance, namely, we can extract from $x^1, \dots, x^K$ two points yielding a solution to \ndtucker.
Mechanically, this is implemented by a set of $N$ agents, one for each dimension $i$, enforcing that $L(i)$ is close to zero for a specific label $i$.

\item The final -- and crucial -- complication that significantly differentiates these reductions from more standard \ppad/-hardness reductions is the presence of \emph{stray cuts}. Indeed, the construction we have described works perfectly assuming that there are $N$ cuts in the input region $I$. However, nothing forces these $N$ cuts
to be made in the input region. If there are less than $N$ cuts in the input region, then we think of the missing cuts as having become \emph{stray cuts} that can interfere with the rest of the construction, in particular the various circuit simulations. Indeed, a stray cut can essentially destroy the output of one of the circuit simulations by occurring in the corresponding region $C^k$. Fortunately, this is easy to fix by taking enough additional samples and copies of the circuit.

The more problematic -- and conceptually important -- interference caused by stray cuts is that a single stray cut can influence any fraction of the circuits (for example, half of them) by ``changing their perception of whether a bit is $1$ or $0$.'' Intuitively, any solution $S = (R^+,R^-)$ of $\CHinstancenoeps$ remains a solution if we flip $+$ and $-$, i.e., if we let $S' = (R^-,R^+)$. This symmetry has the following important consequence: in order to perform a logical operation such as AND, which does not commute with bit-flipping (i.e., $\textup{AND}(\lnot b_1,\lnot b_2) \neq \lnot \textup{AND}(b_1,b_2)$), the circuit needs to be given access to some ground-truth value. This ground-truth essentially helps the circuit differentiate between bits $1$ and $0$, and thus allows it to implement logical gates such as AND. If there are no stray cuts, then it is not too hard to ensure that all the copies of the circuit see the same ground-truth. But a single stray cut can change the ground-truth perception of half the circuits. By a careful construction, it is possible to ensure that if circuit $C^k$'s perception of the ground-truth is altered by a stray cut, then it will output $- \lambda(\overline{x^k})$ instead of $\lambda(x^k)$. Furthermore, the construction ensures that if there are less than $N$ cuts in $I$, and thus at least one stray cut, then all the correctly decoded points amongst $x^1, \dots, x^K$ lie on the boundary of $G$. Using the antipodality conditions of $\lambda$, it follows that $- \lambda(\overline{x^k}) = \lambda(x^k)$ for all valid points $x^k$, and thus the difference in ground-truth between the circuit copies does not matter anymore.

In a certain sense, stray cuts are necessary for any reduction proving \ppa/-hardness for the problem. For example, if there was some trick to enforce that $N$ cuts lie in $I$ in the reduction above, then the reduction would not have made use of the antipodality condition of $\lambda$. This is not possible, since we could then reduce from a circuit $\lambda$ that has no solution, but $\CHinstancenoeps$ always has a solution. More generally, it can be shown that any reduction where each cut has its own disjoint reserved region (where it must lie) can only prove \ppad/-hardness at best. Indeed, in that case those instances can be reduced to the problem of finding a Brouwer fixed point.
\end{itemize}

Our reduction follows this basic template, but requires overcoming various
challenges in order to obtain a constant $\eps$.

\paragraph{\bf Challenge 1: Encoding Tucker solutions.}
While the setup described above is sufficient to obtain hardness for a
polynomially small $\eps$, the encoding of the Tucker solutions fails when one
considers a constant $\eps$. Specifically, in the computation of $L(i)$, note
that most of the terms will be zero, corresponding to points that do not have
label $i$. When $K$ is chosen to be polynomially large, as it is in all prior works,
then the values of $L(i)$ become polynomially small. This does not cause issues
when $\eps$ is also polynomially small, as one can still distinguish $L(i)$
being close to zero, and $L(i)$ being far from zero. But when $\eps$ is a
constant we lose that power, and the reduction breaks.

One idea is to try to get away with only a constant number $K$ of samples and circuit copies. Indeed, prior work by \citet{Rubinstein18-Nash-inapproximability} in the context of \ppad/ has succeeded in performing the so-called averaging trick with only a constant number of samples. Unfortunately, there seems to be a fundamental obstacle to this kind of approach here: there are up to $N$ stray cuts that can ``destroy'' up to $N$ circuit copies, and thus any constant number $K$ of copies will not be enough.

To address this, we define a new version of the Tucker problem where the labels have more expressive power. We call this new problem \ndritucker. Briefly, this
is a variant of \ndtucker in which each point is assigned either $+1$ or $-1$ for
\emph{every} dimension $i$, and thus each point has $N$ labels. That is, the
function $\lambda : G \rightarrow \{-1, +1\}^N$ now returns a vector, such that
$[\lambda(x)]_i$ tells us whether the point has label $+1$ or $-1$ in dimension $i$. The antipodality condition on the boundary can again be formulated as $\lambda(\overline{x}) = - \lambda(x)$ (for $x \in \partial G$).
A solution is a set of $N$ points $x^1, x^2,
\dots, x^N$, that are all within $L_\infty$ distance 1 of each other, and that cover
all labels, meaning that for each dimension $i \in [N]$ there exists a point $x^{\ell_1}$ with
$[\lambda(x^{\ell_1})]_i = +1$ and a point $x^{\ell_2}$ with $[\lambda(x^{\ell_2})]_i = -1$. See
\cref{def:ndritucker} for the formal definition.

Intuitively, in \ndritucker the label $\lambda(x)$ carries much more information than in \ndtucker. Indeed, in a certain sense, the label at some point $x$ now has to pick a direction in \emph{each} dimension $i \in [N]$, and cannot remain ``neutral'' in some dimension. This is exactly what our reduction to $\eps$-\conhalv requires.

We show that \ndritucker is \ppa/-complete, even when the
side-length of the grid is equal to $8$ in all dimensions. We then use
\ndritucker in the reduction to $\eps$-\conhalv. This averts the
problems mentioned above, since now each sum $L(i)$ consists of summands that
are $+1$ and $-1$, and thus $L(i)$ will not be (constantly) close to zero unless both $+1$ and $-1$ appear as labels in dimension $i$.

To show hardness for \ndritucker, we reduce from \twodtucker. We first show
that \twodtucker reduces to \twodritucker by a fairly direct reduction that maps
each of the labels $-2, -1, 1, 2$ from \twodtucker to one of the four
possible vector labels in \twodritucker. Such a simple mapping is not possible
in higher dimensions, however, and so we then use the hardness of \twodritucker
to show hardness for \ndritucker. Here we use a careful adaptation of
the \emph{snake embedding} idea that was used to reduce \twodtucker to \ndtucker in Work 2.
This construction allows us to decrease the width of one of the dimensions by a constant fraction, 
by introducing a new dimension (of small width) and folding the instance within this new
dimension (see \cref{fig:snake-emb}). While this type of embedding has
been used in the past, a fresh construction is needed in our case to deal with
the fact that all points have $N$ labels in an \ndritucker instance.

The \ppa/-hard instances of \twodtucker have exponential width in both
dimensions. 
Repeatedly applying the snake embedding allows us to reduce this to an instance
in which all dimensions
have width $8$. As it turns out, in our final reduction to $\eps$-\conhalv,
the constant 
width of the instance is not strictly necessary (an instance with polynomial widths would
suffice), but we believe that the hardness for constant width may have
applications elsewhere.

\paragraph{\bf Challenge 2: Sampling.}

With this more powerful Tucker problem in hand, one could hope to obtain hardness for constant $\eps$ by simply replacing \ndtucker by \ndritucker in the reduction of Work 3. Unfortunately, there is another point in that reduction that relies on inverse polynomial $\eps$: the sampling. Indeed, that work makes use of arithmetic gates and the so-called equi-angle sampling technique, which has been used in the past to prove \ppad/-hardness for the Nash equilibrium problem \citep{ChenDT09-Nash}. Unfortunately, this sampling technique cannot be combined with constant-error-arithmetic gates. Since we have to take a polynomial number of samples (recall that the stray cuts force us to do this), and the error in each gate is constant, we will not obtain enough distinct samples to ensure that most of them are correctly decoded.

In order to overcome this obstacle, we switch to using Boolean gates, instead of arithmetic gates, like the two original works (Works 1 and 2), while keeping all the other major simplifications introduced by Work 3. Thinking in terms of Boolean gates allows us to construct a very simple, yet very powerful sampling gadget. We subdivide the input region $I$ into subregions $I_1, \dots, I_N$, one for each dimension. The idea is that the $i$th coordinate of the encoded points will be extracted from $I_i$. Next, we subdivide $I_i$ into $7N$ subregions $I_{i,1}, \dots, I_{i,7N}$. We essentially read one bit from each of those $7N$ subregions and interpret the resulting bitstring as the unary representation of a number in $[7N]$. Then, this number is scaled down to lie in $[8]$, in order to correspond to a coordinate in $G = [8]^N$. The crucial point is that we read the coordinate in unary representation and with more precision than actually needed. Thus, even if $N$ bits fail, the final number in $[8]$ will move by at most $1$.

With this simplified sampling technique in hand, it is now possible to reduce to $\eps$-\conhalv for some constant $\eps > 0$.

\paragraph{\bf Challenge 3: Optimizing $\boldsymbol{\eps}$.}

Our final challenge is to push the reduction technique introduced in Works 1 and 2, and simplified in Work 3, to its limits, by trying to obtain hardness for the largest possible value of $\eps$. This effort results in a streamlined reduction that still follows the high-level structure presented above, but where each individual component is as lightweight as possible. Some note-worthy points are:
\begin{itemize}
    \item Switching to Boolean gates, which was very useful to overcome the previous challenge, now becomes a necessity when one is interested in obtaining large constant values of $\eps$. In particular, with the new sampling approach introduced above, the width of \ndritucker does not limit how much we can increase $\eps$. In other words, improving the \ppa/-hardness of \ndritucker to grids of width less than $8$ would not yield an improvement to the $\eps$ we obtain.
    \item Our reduction ends up only using two types of gates: NOT and NAND. Each of these two gates can be implemented by a single agent. The use of NAND instead of AND is not significant, but just for convenience (AND would require creating a NAND gate and then using a NOT gate on its output).
    \item The natural construction of the NAND gate requires $\eps < 1/7$. In order to improve this to $\eps < 1/5$, we eliminate one of the key components introduced in Work 3, the so-called constant creation region, and replace it by an ad-hoc argument which involves arguing about the parity of the number of cuts. This kind of argument is more reminiscent of Works 1 and 2.
\end{itemize}
Putting all these optimizations together, we obtain the reduction presented in \cref{sec:main-red}, which proves \ppa/-hardness of $\eps$-\conhalv for any constant $\eps < 1/5$. The construction also provides a satisfying explanation for why we cannot go above $1/5$ with current techniques. Indeed, it turns out that for NOT gates $\eps < 1/3$ would suffice, but it is the NAND gates which require $\eps < 1/5$. Other parts of the reduction would also work with $\eps < 1/3$. Thus, the NAND gates are clearly identified as the bottleneck for improving $\eps$. More generally, it can be seen that any gate that combines two bits into one in some non-trivial way (e.g., not just copying the first input bit), will require $\eps < 1/5$. Nevertheless, this limitation could be lifted if the reduction was able to handle more than $N$ stray cuts. None of the existing works provide a way to handle this, since all of them crucially rely on there being at most $N$ stray cuts. Indeed, if there are $N+1$ stray cuts, then we can no longer argue that if a stray cut affects our circuits, then we are on the boundary of $G$.

\section{Hardness of \ritucker}

In this section we introduce a new problem, \ndritucker, and we show that it is
\ppa/-hard for any $N \geq 2$. Our reduction from \ndritucker to $\eps$-\conhalv
in \cref{sec:main-red} will also show that the problem is in \ppa/, and
hence the problem is \ppa/-complete.

\subsection{Useful Terminology and Auxiliary Results}
We begin by introducing some notation.
Consider points $\vb{z}_1, \dots, \vb{z}_r$ and a labelling $\lambda$, such that for any $j \in [r]$ we have $\lambda(\vb{z}_{j}) \in \{- 1, +1\}^N$. We say that $\vb{z}_1, \dots, \vb{z}_r$ \textit{cover all labels} if for all $i \in [N]$ and $v \in \{-1, +1\}$ there exists a $j \in [r]$ with $[\lambda(\vb{z}_j)]_i = v$.
Consider an $N$-dimensional grid $[m_1] \times \dots \times [m_N]$ of  points  and a labelling $\lambda : [m_1] \times \dots \times [m_N] \to L$ for some co-domain $L$. The \textit{antipodal point} of a point $\vb{x} = (x_1, \dots, x_N)$ that lies on the boundary of the grid (i.e., $x_i = 1$ or $x_{i} = m_i$ for some $i$) is the point $\overline{\vb{x}} := (m_1 - x_1 + 1, \dots, m_N - x_N + 1)$. We say that the labelling satisfies \textit{antipodality} if $\lambda(\overline{\vb{x}}) = - \lambda(\vb{x})$ for every $\vb{x}$ on the boundary.

We now present an auxiliary lemma that will be useful in this and the following section.
\begin{lemma}\label{lem:r-to-N}
    Consider $r \geq N+1$ points $\vb{z}_1, \dots, \vb{z}_{r}$ and a labelling $\lambda$, such that for any $j \in [r]$ we have $\lambda(\vb{z}_{j}) \in \{- 1, +1\}^N$. If these points cover all labels, then there exists a subset of these points of size at most $N$ that covers all labels. Furthermore, we can recover these at most $N$ points in polynomial time.
\end{lemma}

\begin{proof}
    Consider the set $T := \{ \vb{z}_1, \dots, \vb{z}_{r} \}$.
    We will first show that there exists a multiset of $T$ with cardinality $N+1$, namely, $\vb{z}^*_1, \dots, \vb{z}^*_{N+1}$, that covers all labels. Consider an arbitrary point from $T$ to serve as our desired $\vb{z}^*_1$, without loss of generality $\vb{z}_1$. Its label is $\lambda(\vb{z}_1)$. Then, find a point $\vb{z}^*_2 \in T$ such that $[\lambda(\vb{z}^*_2)]_1 = -[\lambda(\vb{z}_1)]_1$. Next, find a point $\vb{z}^*_3 \in T$ such that $[\lambda(\vb{z}^*_3)]_2 = -[\lambda(\vb{z}_1)]_2$. Similarly, for $j \in \{4, 5, \dots, N+1\}$ find a point $\vb{z}^*_j \in T$ such that $[\lambda(\vb{z}^*_j)]_{j-1} = -[\lambda(\vb{z}_1)]_{j-1}$. Since $T$ covers all labels, this procedure is well-defined. The multiset $\{ \vb{z}^*_1, \dots, \vb{z}^*_{N+1} \}$ we thus obtain has cardinality $N+1$, and covers all labels.
    
    Now consider the distinct elements of the aforementioned multiset, i.e., the set $S := \{ \vb{z}^*_1, \dots, \vb{z}^*_{N+1} \}$. If $|S| \leq N$, then we are done. It remains to handle the case where $|S| = N+1$.
    For the sake of contradiction, assume there is no $N$-subset of $S$ that satisfies the claim of the lemma. Let us create all $N$-subsets of $S$ as follows:
    \begin{align*}
        S_{j} = S \setminus \{\vb{z}^*_j\}, \quad j \in [N+1].
    \end{align*}
    Now consider the function $f : \{ S_1, \dots, S_{N+1} \} \to [N]$ defined as:
    \begin{align*}
        f(S_{j}) = \min \{i \in [N] : [\lambda(\vb{z}^*_{k})]_{i} = [\lambda(\vb{z}^*_{\ell})]_{i} \quad \forall k,\ell \neq j\}.
    \end{align*}
    Note that $f$ is well-defined since, by assumption, every $S_j$ has such a minimum index. Then, by the pigeonhole principle, this function maps two elements $S_{j'}, S_{j''}$ of its domain to the same value $i' \in [N]$, i.e. $f(S_{j'}) = f(S_{j''}) = i'$. Therefore, the set of points $S_{j'} \cup S_{j''}$ also has the property that the $i'$-th coordinate of all its points' labels has the same value in $\{ -1, +1 \}$. But by definition of the $S_{j}$'s, we have $S_{j'} \cup S_{j''} = S$, thus our initial assumption that the points of $S$ cover all labels does not hold (the label-coordinate $i'$ is not covered), which is a contradiction.
    
    To find the set $S$ we need to check $r$ many points in the worst case. To recover from $S$ the required $N$-subset we need to check at most all of its $N+1$ many $N$-subsets. Considering the polynomial time that the labelling circuit $\lambda$ needs in order to provide us with the requested labels of the points in the above procedure, we conclude that the overall time to recover the desired $N$ points is polynomial. 
\end{proof}

We now formally define \ndritucker.
\begin{definition}[\ndritucker]
\label{def:ndritucker}
An instance of \ndritucker consists of a labelling $\lambda : [m_1] \times \dots \times [m_N] \to \{- 1, +1\}^N$ (represented by a Boolean circuit) that satisfies antipodality. A solution consists of $N$ points\footnote{The points do not need to be distinct.} $\vb{z}_1, \dots, \vb{z}_N$ that cover all labels, and such that $|| \vb{z}_{j} - \vb{z}_{k}||_{\infty} \leq 1$ for all $j,k \in [N]$.
\end{definition}

The following theorem states that \ndritucker always has a solution.

\begin{theorem}
Let us have an $N$-dimensional grid $[m_1] \times \dots \times [m_N]$ of points and a labelling $\lambda : [m_1] \times \dots \times [m_N] \to \{- 1, +1\}^N$ that satisfies antipodality. Then, there exist $N$ points $\vb{z}_1, \dots, \vb{z}_N$ that cover all labels such that $|| \vb{z}_{j} - \vb{z}_{k}||_{\infty} \leq 1$ for all $j,k \in [N]$.
\end{theorem}

\begin{proof}
    The proof of existence is indirect, and comes from the proof of \ppa/-inclusion of \ndritucker presented in \cref{sec:main-red}. In the aforementioned section we prove that \ndritucker reduces to $\eps$-\conhalv for any constant $\eps < 1/5$. And by the fact that $\eps$-\conhalv is in \ppa/ \citep{FRG18-Consensus}, we get the required inclusion. Alternatively, one could also reduce the problem to some version of Borsuk-Ulam by taking an appropriate continuous interpolation of the labelling.
\end{proof}

\subsection{The Reduction}

In this section we show the following theorem.

    \begin{theorem}\label{thm:ST-hardness}
    	\ndritucker is \ppa/-complete even when $m_i = \wt$ for all $i \in [N]$.
    \end{theorem}

The remainder of this section is devoted to proving this theorem.
The reduction consists of two steps. 
We first show that the 2D
version of the problem is hard, and then we show that hardness for the 2D case
implies hardness for higher dimensional instances. 
The following theorem shows hardness for \twodritucker via a direct reduction from
\twodtucker.

\begin{theorem}
	\twodritucker is \ppa/-complete.
\end{theorem}

\begin{proof}
	\twodtucker is known to be \ppa/-complete \citep{AisenbergBB20-2D-Tucker}. We will reduce this problem to \twodritucker straightforwardly by just translating the labelling $\lambda_{T} : [m] \times [m] \to \{\pm 1, \pm 2\}$ of the former to the labelling $\lambda_{ST} : [m] \times [m] \to \{-1, +1\}^2$ of the latter as follows. For any point $\vb{x}$, if $\lambda_{T}(\vb{x}) = +2$ then $\lambda_{ST}(\vb{x}) = (+1, +1)$, if $\lambda_{T}(\vb{x}) = -2$ then $\lambda_{ST}(\vb{x}) = (-1, -1)$, if $\lambda_{T}(\vb{x}) = +1$ then $\lambda_{ST}(\vb{x}) = (+1, -1)$, and if $\lambda_{T}(\vb{x}) = -1$ then $\lambda_{ST}(\vb{x}) = (-1, +1)$. By definition of the problems, it is immediate that their sets of solutions are identical. 
	
	By reversing the above translation of the labelling, i.e. turning $\lambda_{ST}$ to $\lambda_{T}$ using the same mapping, we get a reduction from \twodritucker to \twodtucker, and hence, the former problem's membership to \ppa/.
\end{proof}

    \paragraph{\bf Overview of the reduction from \twodritucker to \ndritucker.}
	We will reduce \twodritucker with width $m = 2^M$ to \ndritucker with width 8 for some appropriate value of $N = O(M)$. The reduction is, in essence, a careful application of the well-known snake embedding technique \citep{ChenDT09-Nash,FRG19-Necklace} which was used to reduce \twodtucker to a \tucker problem of higher dimension. We have to carefully apply the latter technique for our problem since we need to make sure that no artificial solutions are introduced in the ``folding'' process, and that the folded $k$-dimensional instance in each step is a proper $(k+1)$-dimensional instance, meaning that it preserves antipodality. As a final step, we ensure that all dimensions have width exactly \wt.
	
	The snake embedding technique starts from the $2^M \times 2^M$ \twodritucker instance and at each step performs a ``folding'' on some dimension, decreasing its width to roughly $1/3$ of its size, while creating a new dimension of width \wt. In this way, in roughly $2 \cdot \log_{3}2^M$ foldings we have created an equal amount of extra dimensions of width at most \wt. In general, given a $k$D-\ritucker instance for $k \geq 2$, by performing a folding on its $i$-th dimension, we create a $(k+1)$D-\ritucker instance with new width $m_i' \leq \ceil{\frac{m_i}{3}} + 4$ and an extra $(k+1)$-st dimension of width \wt. Finally, we perform two extra foldings to ensure that our initial dimensions $1$ and $2$ have also width \wt.
	
	\begin{figure}
		\centering
		\includegraphics[scale=1.00]{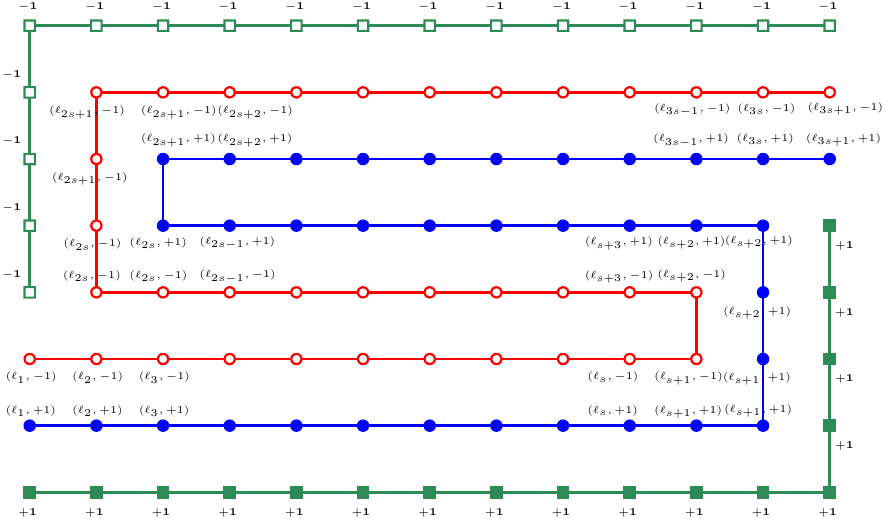}
		\caption{A slice of the $(k+1)$D-\ritucker instance $I^*_{ST}$ for fixed coordinates of all dimensions other than the folding dimension $i$ ($x$-axis) and $k+1$ ($y$-axis). $\ell_{j}$ is the label $\lambda$ of the point with coordinate $j$ in dimension $i$ in the $k$D-\ritucker instance $I_{WST}$. By $\mathbf{+1}$ and $\mathbf{-1}$ we denote the $(k+1)$-dimensional vector (label) with all entries $+1$ and $-1$ respectively.}\label{fig:snake-emb}
	\end{figure}

	\paragraph{\bf Making the width of $\boldsymbol{k}$D-\ritucker suitable for folding.}
	We now describe a general step of the snake embedding, that is, a step where we are given a $k$D-\ritucker instance and we fold it into a $(k+1)$D-\ritucker instance. Pick a dimension of $k$D-\ritucker, without loss of generality $i \in [k]$, that has maximum width $m_{i} > 8$, if any. We will call this the \textit{folding dimension}, and for some $d \in [m_i]$, let us call \textit{$d$-th ray} the set of points of the grid that have coordinate $d$ in that dimension. According to our folding procedure, the $i$-th dimension will have to be of width of the form $3 \cdot s + 1$, for some natural number $s$. Therefore, for width $m_i$ that is not of the aforementioned size, we have to add extra copies of rays in order to bring it to the required width. When we refer to adding copies of rays we mean that we copy sets of points together with their labels. In order to preserve antipodality we have to take care of how many copies of the $1$-st and $m_i$-th ray we will add. To achieve this, instead of adding one ray when needed, we can attach four extra rays, namely, two left of coordinate $1$ and two right of coordinate $m_i$. Let us call the initial $k$D-\ritucker instance $I_{ST}$ and the one with proper width $I_{WST}$.

	Let us use the following set of rules that depend on the size of $m_i$ and preserve antipodality:
	\begin{itemize}
	    \item If $m_i = 3 \cdot s' + 2$ we add one copy of the $1$-st ray left of the $1$-st ray and one copy of the $m_i$-th ray right of the $m_i$-th ray.
	    \item If $m_i = 3 \cdot s' + 1$ we do not need to add any ray.
	    \item If $m_i = 3 \cdot s'$ we add two copies of the $1$-st ray left of the $1$-st ray and two copies of the $m_i$-th ray right of the $m_i$-th ray.
	\end{itemize}
	The above additions of rays ensure that the width of the $i$-th dimension of $I_{WST}$ is $3 \cdot s + 1$ for some $s \in \naturals^*$. 
	
	\paragraph{\bf Copying and folding $\boldsymbol{I_{WST}}$, and creation of the $\boldsymbol{(k+1)}$D-\ritucker instance.}
	In essence, we create two identical (up to the turning points) copies of $I_{WST}$ that we glue together and fold in a snake-like shape. 
	Let us call \textit{bottom snake} the bottom layer of $I_{WST}$ as appears in \cref{fig:snake-emb} (blue/shaded-circle layer), and \textit{top snake} the top layer of $I_{WST}$ (red/hollow-circle layer). Then we need to take care of the turns of $I_{WST}$ so that they do not introduce artificial solutions. To achieve this, it suffices that the bottom snake is formed by copying the $(s+1)$-st and $(s+2)$-nd rays two times, and the top snake is formed by copying the  $2s$-th and $(2s+1)$-st rays two times. Then, the folding in the $i$-th dimension of the bottom and top snakes is as demonstrated in \cref{fig:snake-emb}.

	Next, we add extra rays below the bottom snake and above the top snake (green/shaded-squares and green/hollow-squares respectively in \cref{fig:snake-emb}). In particular, we add rays whose coordinates in the $i$-th and $(k+1)$-st dimensions are
	\begin{itemize}
	    \item $(j, 1)$ for all $j \in \{ 1, \dots, s+2 \}$ and $(s+3, m)$ for all $m \in \{ 1, \dots, 5 \}$, which consist the \textit{bottom cap}, and symmetrically,
	    \item $(j, 8)$ for all $j \in \{ 2, \dots, s+3 \}$ and $(1, m)$ for all $m \in \{ 4, \dots, 8 \}$, which consist the \textit{top cap}.
	\end{itemize}
	Let us call $I^*_{ST}$ the resulting $(k+1)$D-\ritucker instance. From the described folding procedure, we conclude that by folding the $i$-th dimension of $I_{WST}$ for which $m_i = 3 \cdot s + 1$, we generate $I^*_{ST}$ which has an extra $(k+1)$-st dimension of width \wt and its $i$-th dimension has now width $m'_{i} = s + 3$ (see \cref{fig:snake-emb}). 
	
	\paragraph{\bf Mapping points of $\boldsymbol{I^*_{ST}}$ to points of $\boldsymbol{I_{WST}}$.}
	From the construction so far, we can determine a surjection of points of the bottom and top snakes in $I^*_{ST}$ to points in $I_{WST}$. We only need such a surjection because, as we will show later, no point of the bottom or top cap can participate in a solution of $I^*_{ST}$. We will map the ray $(j,m)$ corresponding to the coordinates of the $i$-th and the $(k+1)$-st dimensions of $I^*_{ST}$ to the $t$-th ray corresponding to the coordinate of the $i$-th dimension of $I_{WST}$. When we say that we map ray $r_1$ to ray $r_2$ we imply that any point in $r_1$ with fixed coordinates in the $k-1$ of its dimensions maps to the point of $r_2$ with the same coordinates of these $k-1$ dimensions. The surjection is as follows.
	\begin{itemize}
	    \item $(j, m)$ for $j \in \{ 1, \dots, s+1 \}$ and $m \in \{ 2, 3 \}$ maps to $t = j$.
	    \item $(s+2, m)$ for $m \in \{ 2, 3 \}$ maps to $t = s+1$.
	    \item $(s+2, m)$ for $m \in \{ 4, 5 \}$ maps to $t = s+2$.
	    \item $(j, m)$ for $j \in \{ 3, \dots, s+1 \}$ and $m \in \{ 4, 5 \}$  maps to $t = 2s+3-j$.
	    \item $(2, m)$ for $m \in \{ 4, 5 \}$ maps to $t = 2s$.
	    \item $(2, m)$ for $m \in \{ 6, 7 \}$ maps to $t = 2s+1$.
	    \item $(j, m)$ for $j \in \{ 3, \dots, s+3 \}$ and $m \in \{ 6, 7 \}$ maps to $t = 2s-2+j$.
	\end{itemize}

	\paragraph{\bf Labelling $\boldsymbol{I^*_{ST}}$.}
	Having specified the structure of the $(k+1)$D-\ritucker instance, we have to determine the labels of its points. By the construction described in the previous paragraph, the added rays determine the first $k$ label-coordinates of the points in the bottom and top snakes. The $(k+1)$-st label-coordinate of each point in the two snakes is determined as follows: for the bottom snake its value is $+1$ and for the top snake its value is $-1$. Finally, for all points of the bottom cap the label is $\mathbf{+1}$, i.e., all label-coordinates get value $+1$, and similarly, for all points of the top cap the label is $\mathbf{-1}$.
	
	\paragraph{\bf Correctness of the reduction.}
    So far we have made sure that at each step of the folding procedure the $k$-dimensional instance $I_{ST}$ at hand and also its modified version $I_{WST}$ will be proper $k$D-\ritucker instances. Now we will prove correctness of the reduction by showing that every solution of the final \ndritucker instance corresponds to a solution in the initial \twodritucker instance. We will show this by proving that at every step $k \geq 2$ of the folding procedure, every solution of the $(k+1)$D-\ritucker instance $I^*_{ST}$ corresponds to a solution of the $k$D-\ritucker instance $I_{WST}$.

    Suppose that $S'= \{ \vb{z}'_{1}, \dots, \vb{z}'_{k+1} \}$ is a solution to $I^*_{ST}$. Let us prove the following claim.
    
    \begin{claim}
        No point of $S'$ belongs to the bottom or top cap.
    \end{claim}
    
    \begin{proof}
        Let us first consider the bottom cap. For the sake of contradiction, suppose that a point of $S'$ belongs to the bottom cap. Then, since $|| \vb{z}'_{\ell} - \vb{z}'_{m} ||_{\infty} \leq 1$ for all $\ell, m \in [k+1]$, all other points belong to the union of the bottom snake and the bottom cap. But then, the $(k+1)$-st label-coordinate of all points in $S'$ is $+1$, contradicting the property of $S'$ that its points cover all labels.
        
        Similarly, if one of the points from $S'$ belonged to the top cap, then the labels of all points in $S'$ would have their $(k+1)$-st coordinate equal to $-1$, a contradiction.
    \end{proof}
    
    By the labelling in the folding we have specified earlier, any solution $S'$ has to include at least one point of the bottom snake and at least one point of the top snake, otherwise their $(k+1)$-st label-coordinates would be the same - either $+1$ or $-1$ - contradicting the property of covering all labels. Let us call a point $\vb{z}'$ of $I^*_{ST}$ a \textit{bottom corner point} if it belongs to one of the copied rays of the bottom snake. Similarly, let us call it a \textit{top corner point} if it belongs to one of the copied rays in the top snake. 
    
    Consider for each $\vb{z}'_{j} \in S'$ the point $\vb{z}^*_{j}$ of $I_{WST}$ to which it is mapped according to the respective paragraph above. Let $S^* = \{\vb{z}^*_1, \dots, \vb{z}^*_{k+1}\}$ be the set of the corresponding $k$-dimensional points of $I_{WST}$.
    
    Let the labellings of $I^*_{ST}$ and $I_{WST}$ be denoted by $\lambda'$ and $\lambda$, respectively. 
    According to the mapping of points of the former to points of the latter instance we have defined in the respective paragraph above, it is immediate that if we have $|| \vb{z}'_{\ell} - \vb{z}'_{m} ||_{\infty} \leq 1$ for every $\vb{z}'_{\ell}, \vb{z}'_{m} \in S'$, then $|| \vb{z}^*_{\ell} - \vb{z}^*_{m} ||_{\infty} \leq 1$ for every $\vb{z}^*_{\ell}, \vb{z}^*_{m} \in S^*$. Furthermore, the labellings that we have defined above copy the labellings of the respective points that we map from $I^*_{ST}$. Therefore, if $S'$ covers all labels in $\lambda'$ (which has $k+1$ coordinates) then $S^*$ covers all labels in $\lambda$ (which has $k$ coordinates). Observe that $|S^*| \leq k+1$.
    If $|S^*| = k+1$, from \cref{lem:r-to-N} we get that there is a $k$-subset of $S^*$ which is a solution to $I_{WST}$. Furthermore, this solution can be found in polynomial time from $S^*$.

    \paragraph{\bf Putting everything together.}
    
    We are now ready to finish the proof of \cref{thm:ST-hardness}.
    By repetitions of the cycle $I_{ST} \to I_{WST} \to I^*_{ST}$, we fold the widest dimension of the $k$D-\ritucker starting from $k=2$ until all dimensions' widths are at most \wt. As we showed above, this is guaranteed to happen after linearly many foldings, by the fact that every folding dimension $i$ has width $m_{i} = 3 \cdot s + 1 > \wt$ for some $s \in \naturals^*$, and after the folding its width is $m'_i = s + 3$, while a $(k+1)$-st dimension of width \wt has been created (see \cref{fig:snake-emb}). 
   
    We will now show the final step that makes all the dimensions' widths exactly \wt. Recall that we have started from the $2^M \times 2^M$ \twodritucker instance, and by repeatedly folding dimensions $1$ and $2$ we have generated extra dimensions of width \wt. Therefore, the only dimensions we need to take care of are the aforementioned two. Let us consider dimension $i \in \{1, 2\}$ and recall that after a folding, the folded dimension's width is $s + 3$ for some $s \in \naturals^*$, while initially it was $3 \cdot s + 1$. Let us denote by $m^{t}_{i}$ the width of the dimension after $t$ foldings. As mentioned earlier, our folding technique reduces the size of the folding dimension $i$ from $m^{t-1}_i$ to $m^{t}_i \leq \ceil{\frac{m^{t-1}_i}{3}} + 4 \leq \frac{m^{t-1}_i}{3} + 5$, for any $t \in \naturals^*$. This recursion induces the following inequality: 
    \begin{align*}
        m^{t}_i \leq \frac{m^{0}_i}{3^{t}} + \frac{15}{2} \cdot \left( 1 - \frac{1}{3^t} \right).
    \end{align*}
    Therefore, we can ensure that the left-hand side is at most \wt by forcing the right-hand side to be at most \wt, which can be achieved in at most $\ceil{ \log_{3}(2 \cdot m^{0}_{i} - 15) }$ steps. Before the first folding (i.e. after making the width proper for folding), the width of dimension $i \in \{ 1, 2 \}$ of our initial \twodritucker instance will be $m^{0}_{i} \leq 2^M + 4$, therefore after some number $t^* \leq \ceil{ \log_{3}(2^{M+1} - 7) }$ of foldings we have $m^{t^*}_i \leq 8$, at which point we stop.
    
    If $m^{t^*}_i = 8$ then we are done. If $m^{t^*}_i < 8$, notice that it will necessarily be an even number. That is because, for any given $t \in \naturals^*$, if $m^{t-1}_i = 3 \cdot s + 1$ is even then $m^{t}_i = s + 3$ is even, and we have started with $m^{0}_i = 2^M + 2 \cdot p$ for some $p \in \{ 0, 1, 2 \}$ (to make sure that it is of the form $3 \cdot s + 1$ before the folding, as described in the respective paragraph above). Since $m^{t^*}_i$ is even, let us copy its $1$-st ray $(16 - m^{t^*}_i)/2$ times to the left and its $m^{t^*}_i$-th ray $(16 - m^{t^*}_i)/2$ times to the right, and create a modified instance with $m'_{i} = 16$ by following the procedure described in the respective paragraph above. Recall that this procedure ensures that the modified instance  preserves antipodality. Now perform a final folding which will bring the width from $16 = 3 \cdot 5 + 1$ to $5 + 3 = 8$. 
    
    Finally, inclusion of \ndritucker in \ppa/ comes from the reduction of \ndritucker to $\eps$-\conhalv for any constant $\eps < 1/5$ presented in  \cref{sec:main-red}. As shown by \citet{FRG18-Consensus}, $\eps$-\conhalv is in \ppa/, which implies the required inclusion.

\section{Main Reduction}\label{sec:main-red}

In this section, we prove our main result, \cref{thm:main-result}. Namely, for any constant $\eps < 1/5$, we present a polynomial-time reduction from \ndritucker to $\eps$-\conhalv. In \cref{sec:3-block} we explain how our reduction can be modified to work with 3-block uniform valuations, thus proving \cref{thm:main-result-3-block}.

Fix any $\eps \in [0,1/5)$. Let $\lambda$ be an instance of \ndritucker, i.e., $\lambda: [8]^N \to \{-1,+1\}^N$ is provided as a Boolean circuit. We use $\sz(\lambda)$ to denote the representation size of the Boolean circuit $\lambda$. Note that, in particular, $\sz(\lambda) \geq N$. We show how to construct an instance $\CHinstance$ of $\eps$-\conhalv in time polynomial in $\sz(\lambda)$, such that from any solution of $\CHinstance$ we can extract in polynomial time a solution to $\lambda$.

\subsection{Pre-processing}

\paragraph{\bf Construction of the modified circuit $\boldsymbol{\mlambda}$.} The first step of the reduction is to construct a slightly modified version of $\lambda$, which will be more convenient to work with. First of all, we will not think of bits as lying in $\{0,1\}$, but, instead, in $\{-1,+1\}$. Here, $-1$ will represent bit $0$ (``False''), and $+1$ will represent bit $1$ (``True'').

With this interpretation in mind, the modified circuit, which we denote by $\mlambda$, is defined as follows. The input to $\mlambda$ consists of $7N^2$ bits, that we think of as a matrix $x \in \{-1,+1\}^{N \times 7N}$. We use $x_{i,j} \in \{-1,+1\}$ to denote the $(i,j)$ entry, and $x_i \in \{-1,+1\}^{7N}$ to denote the $i$th row. The circuit outputs $N$ bits representing a label $\{-1,+1\}^N$. On input $x$, the circuit performs the following computations.
\begin{enumerate}
    \item For each $i \in [N]$, compute
    \begin{equation}\label{eq:def-phi}
    \phi_i(x) := \left\lceil\frac{8N + 1/2 + \sum_{j=1}^{7N} x_{i,j}}{2N}\right\rceil \in [8].
    \end{equation}
    \item Compute and output $\lambda(\phi(x)) \in \{-1,+1\}^N$.
\end{enumerate}
In time polynomial in $\sz(\lambda)$ we construct a Boolean circuit $\mlambda$ that performs these computations, and only uses NOT gates and NAND gates. Note that other logical gates can easily be simulated using these two gates.

Intuitively, the circuit $\mlambda$ does the following. For any $i \in [N]$, $x_i \in \{-1,+1\}^{7N}$ is interpreted as representing a number between $1$ and $8$ with precision roughly $1/N$ (in unary representation). That number is then rounded to obtain an integer $ \phi_i(x) \in [8]$. Why do we use more bits than needed to represent a number in $[8]$? The reason is that this representation is robust to flipping a few bits. Indeed, it is easy to check that flipping up to $N$ bits of $x_i \in \{-1,+1\}^{7N}$ changes the value of $\phi_i(x)$ by at most $1$. As a result, we obtain the following:

\begin{claim}\label{clm:at-most-N-bits}
If $x, x' \in \{-1,+1\}^{N \times 7N}$ are such that for all $i \in [N]$, $x_i$ and $x_i'$ differ in at most $N$ bits, then $\|\phi(x) - \phi(x')\|_\infty \leq 1$.
\end{claim}

\begin{proof}
If $x_i$ and $x_i'$ differ in at most $N$ bits, then
\[\left|\sum_{j=1}^{7N} x_{i,j} - \sum_{j=1}^{7N} x_{i,j}'\right| \leq 2N\]
and as a result $|\phi_i(x) - \phi_i(x')| \leq 1$ by \cref{eq:def-phi}.
\end{proof}

The circuit $\mlambda$ consists of $m$ gates $g_1, \dots, g_m$, where $m \geq N$ and $m \leq \sz(\mlambda) \leq \textup{poly}(\sz(\lambda))$. For each $t \in [m]$, $g_t = (g_{t_1}, g_{t_2},T)$, where $t_1, t_2 \in [t-1] \cup ([N] \times [7N])$ are the inputs to the gate, and $T \in \{\textup{NOT}, \textup{NAND}\}$ indicates the type of gate. Note that an input $g_{t_1}$ to a gate $g_t$ can be of two types: when $t_1 \in [t-1]$, then $g_{t_1}$ is simply another (``earlier'') gate of the circuit; when $t_1 \in [N] \times [7N]$, then $g_{t_1} = g_{(i,j)}$, which we interpret as the $(i,j)$th input to the circuit, i.e., $x_{i,j}$. Note that when $T = \textup{NOT}$, the second input $g_{t_2}$ is ignored. The output of the circuit $\mlambda$ is given by the last $N$ gates, i.e., $g_{m-N+1}, \dots, g_m$.

\subsection{Construction of the Instance}

We now begin with the description of the $\eps$-\conhalv instance $\CHinstance$ that we construct. Instead of working with the interval $[0,1]$, we will describe the construction on an interval $R = [0, \textup{poly}(\sz(\lambda))]$. The valuations of the agents can then easily be scaled down to $[0,1]$.

\paragraph{\bf Input and circuit regions.} The interval $R$ is subdivided into two subintervals: interval $I$ on the left, and interval $C$ on the right. Interval $I$ is called the ``Input region'', while $C$ is called the ``Circuit region''. The interval $I$ is further subdivided into intervals $I_1, I_2, \dots, I_N$ from left to right. Next, each interval $I_i$ is subdivided into intervals $I_{i,1}, \dots, I_{i,7N}$. Finally, each interval $I_{i,j}$ is subdivided into intervals $I_{i,j}^1, \dots, I_{i,j}^{3N}$. Each of those final small intervals has length $1$, i.e., $|I_{i,j}^k| = 1$. Thus, the total length of interval $I$ is $N \cdot 7N \cdot 3N = 21 N^3$.

The interval $C$ is subdivided into intervals $C^1, C^2, \dots, C^{3N}$. We think of each $C^k$ as being associated to a separate ``copy'' of the circuit $\mlambda$. Next, each interval $C^k$ is subdivided into intervals $C^k_1, \dots, C^k_m$, where we recall that $m$ is the number of gates of $\mlambda$. Finally, each interval $C^k_t$ is subdivided into intervals $C^k_{t,\ell}, C^k_{t,c}, C^k_{t,r}, C^k_{t,a}$. The intervals $C^k_{t,\ell}, C^k_{t,c}, C^k_{t,r}, C^k_{t,a}$ have length $1$ each, and are called the \emph{left/center/right/auxiliary} subinterval of $C^k_t$, respectively. Putting everything together, we see that $|C| = 3N \cdot m \cdot 4 = 12mN$ and thus $|R| = |I \cup C| = 21 N^3 + 12mN = \textup{poly}(\sz(\lambda))$.

\paragraph{\bf Agents.} The instance $\CHinstance$ will have exactly $n = 3N \cdot m \cdot 2 + N$ agents. Namely, for each $k \in [3N]$ and $t \in [m]$, there is a \emph{gate agent} $\alpha^k_t$ and an \emph{auxiliary agent} $\beta^k_t$. We think of these agents as ``belonging'' to the interval $C^k_t$. Furthermore, there are also \emph{feedback agents} $\gamma_1, \dots, \gamma_N$. We will define the valuation functions for all these agents below, but first we have to introduce the notion of the value encoded by an interval.

\paragraph{\bf Value of an interval.} Consider any solution $S$ of our instance $\CHinstance$. Then $S = (R^+,R^-)$ is a partition of $R$ into two parts $R^+$ and $R^-$ using at most $n$ cuts. Without loss of generality, we can assume that $S$ has the following property: the right-most end of $R$ lies in $R^+$. Indeed, if this is not the case, then swapping $R^+$ and $R^-$ yields a solution that satisfies this.

In any solution $S = (R^+,R^-)$, we can assign a value in $[-1,1]$ to any interval $J \subset R$, $|J| = 1$, in a natural way:
\[\val_S(J) := \mu(J \cap R^+) - \mu(J \cap R^-)\]
where $\mu$ is the Lebesgue measure on $\mathbb{R}$. When it is clear from the context, we will omit the subscript $S$. We say that the value of an interval $J$ is \emph{pure}, if $\val(J) \in \{-1,+1\}$, i.e., it can be interpreted as a bit.

For $k \in [3N]$, $i \in [N]$, and $j \in [7N]$, we let
\[x^k_{i,j} := \val(I^k_{i,j}) \in [-1,1].\]
Furthermore, for $k \in [3N]$ and $t \in [m]$, we let
\[g^k_t := \val(C^k_{t,c}) \in [-1,1].\]
For convenience, we also define $g^k_t := x^k_{i,j}$, when $t = (i,j) \in [N] \times [7N]$, i.e., when $t$ refers to an input of the circuit, and not a gate.

We think of $x^1, \dots, x^{3N}$ as $3N$ possible inputs to our circuit $\mlambda$. Of course, $\mlambda(x^k)$ is only well-defined if $x^k$ is pure, i.e., if $x^k \in \{-1,+1\}^{N \times 7N}$. We can make the following crucial observations.

\begin{claim}\label{clm:at-most-N-cuts-in-I}
In any solution $S$ where at most $N$ cuts lie in the interior of interval $I$, it holds that, if $x^{k_1}$ and $x^{k_2}$ are both pure, then $\|\phi(x^{k_1}) - \phi(x^{k_2})\|_\infty \leq 1$ and $\lambda(\phi(x^{k_i})) = \mlambda(x^{k_i})$ for $i=1,2$.
\end{claim}

\begin{proof}
The statement $\lambda(\phi(x^{k_i})) = \mlambda(x^{k_i})$ follows by the construction of $\mlambda$. It remains to prove that $\|\phi(x^{k_1}) - \phi(x^{k_2})\|_\infty \leq 1$.
Since the interior of $I$ contains at most $N$ cuts, it follows that for each $i \in [N]$, the interior of the interval $I_i$ contains at most $N$ cuts. As a result, there exists a subset $P_i \subseteq [7N]$ with $|P_i| \geq 7N - N = 6N$ such that for all $j \in P_i$ the interior of interval $I_{i,j}$ does not contain any cuts. This means that for all $j \in P_i$, the intervals $I_{i,j}^{k_1}$ and $I_{i,j}^{k_2}$ have the same value, i.e., $x^{k_1}_{i,j} = \val(I_{i,j}^{k_1}) = \val(I_{i,j}^{k_2}) = x^{k_2}_{i,j}$. Thus, since $|P_i| \geq 6N$, $x^{k_1}_i$ and $x^{k_2}_i$ differ in at most $N$ bits. Since this holds for all $i \in [N]$, the claim follows by \cref{clm:at-most-N-bits}.
\end{proof}

\begin{claim}\label{clm:at-most-N-1-cuts-in-I}
In any solution $S$ where at most $N-1$ cuts lie in the interior of interval $I$, it holds that, if $x^k$ is pure, then $\mlambda(-x^k) = - \mlambda(x^k)$.
\end{claim}

\begin{proof}
Since the interior of $I$ contains at most $N-1$ cuts, there exists $s \in [N]$ such that the interior of $I_s$ does not contain any cuts. As a result, $x^k_{s,j_1} = \val(I^k_{s,j_1}) = \val(I^k_{s,j_2}) = x^k_{s,j_2}$ for all $j_1,j_2 \in [7N]$. By the definition of $\phi$ (\cref{eq:def-phi}), it follows that $\phi_s(x^k) \in \{1,8\}$. Thus, by the boundary conditions of $\lambda$, we obtain that $\lambda(\overline{\phi(x^k)}) = - \lambda(\phi(x^k))$, where $\overline{\phi_i(x^k)} = 9 - \phi_i(x^k)$ for all $i \in [N]$. Since $\lambda(\phi(x^k)) = \mlambda(x^k)$, it remains to show that $\phi(-x^k) = \overline{\phi(x^k)}$.

Fix any $i \in [N]$ and consider $\phi_i(x^k) = q \in [8]$. By the definition of $\phi$ (\cref{eq:def-phi}), it follows that
\[(q-1) \cdot 2N < 8N + 1/2 + \sum_{j=1}^{7N} x^k_{i,j} \leq q \cdot 2N\]
which implies that
\[(8-q) \cdot 2N + 1 \leq 8N + 1/2 - \sum_{j=1}^{7N} x^k_{i,j} < (9-q) \cdot 2N + 1\]
and finally
\[(8-q) \cdot 2N < 8N + 1/2 - \sum_{j=1}^{7N} x^k_{i,j} \leq (9-q) \cdot 2N.\]
But, by the definition of $\phi$ (\cref{eq:def-phi}), this exactly means that $\phi_i(-x^k) = 9-q = 9 - \phi_i(x^k) = \overline{\phi_i(x^k)}$.
\end{proof}

\paragraph{\bf Auxiliary agents.} For $k \in [3N]$ and $t \in [m]$, the auxiliary agent $\beta^k_t$ has a very simple valuation function $v_{\beta^k_t}$: the density function of the valuation has value $1$ in $C^k_{t,a}$, and value $0$ everywhere else. This corresponds to having a block of volume $1$ lying in interval $C^k_{t,a}$. We immediately obtain the following observation.

\begin{claim}\label{clm:cut-in-auxiliary}
For all $k \in [3N]$ and $t \in [m]$ there must be a cut in the interior of $C^k_{t,a}$.
\end{claim}

\begin{proof}
If there is no cut in the interior of $C^k_{t,a}$, then $|v_{\beta^k_t}(R^+) - v_{\beta^k_t}(R^-)| = 1 > \eps$.
\end{proof}

\paragraph{\bf Gate agents: NOT.} Let $t \in [m]$ be such that $g_t=(g_{t_1},g_{t_2},\textup{NOT})$. Then, for any $k \in [3N]$, the goal of gate agent $\alpha^k_t$ is to enforce the corresponding gate constraint, namely $g^k_t = \val(C^k_{t,c}) = \textup{NOT}(g^k_{t_1}) = - g^k_{t_1}$.
The density function of the valuation $v_{\alpha^k_t}$ is constructed as follows: it has value $1/3$ in $C^k_{t,\ell} \cup C^k_{t,r} \cup A_{t_1}$, and value $0$ everywhere else. Here $A_{t_1}$ is defined as
\begin{itemize}
    \item $A_{t_1} = C^k_{t_1,c}$, when $t_1 \in [t-1]$,
    \item $A_{t_1} = I^k_{i,j}$, when $t_1 = (i,j) \in [N] \times [7N]$.
\end{itemize}
Note that $\val(A_{t_1}) = g^k_{t_1}$. See \cref{fig:NOT-gate} for an illustration of the gate.

\begin{figure}
	\centering
	\includegraphics[scale=1.2]{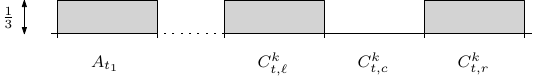}
	\caption{The density function of the valuation of an agent $\alpha_t^k$ implementing a NOT-gate.}\label{fig:NOT-gate}
\end{figure}

\begin{claim}\label{clm:NOT-gate}
For all $t \in [m]$ such that $g_t=(g_{t_1},g_{t_2},\textup{NOT})$, and all $k \in [3N]$, it holds that:
\begin{itemize}
    \item there must be a cut in the interior of $C^k_{t,\ell} \cup C^k_{t,c} \cup C^k_{t,r}$;
    \item if there is exactly one cut in the interior of $C^k_{t,\ell} \cup C^k_{t,c} \cup C^k_{t,r}$, and the input to the gate is pure (namely, $g^k_{t_1} \in \{-1,+1\}$), then the output is pure (i.e., $g^k_t \in \{-1,+1\}$), and $g^k_t = \textup{NOT}(g^k_{t_1})$.
\end{itemize}
\end{claim}

\begin{proof}
Assume, towards a contradiction, that the interior of $C^k_{t,\ell} \cup C^k_{t,c} \cup C^k_{t,r}$ does not contain any cuts. Then, in particular, $C^k_{t,\ell}$ and $C^k_{t,r}$ are both contained in $R^+$ or both contained in $R^-$. This implies that $|v_{\alpha^k_t}(R^+) - v_{\alpha^k_t}(R^-)| \geq 2/3 - 1/3 = 1/3 > \eps$, a contradiction.

Now consider the case where the interior of $C^k_{t,\ell} \cup C^k_{t,c} \cup C^k_{t,r}$ contains exactly one cut. Let $A_{t_1}$ be as defined above. Recall that $g^k_{t_1} = \val(A_{t_1})$ and $g^k_t = \val(C^k_{t,c})$. If $\val(A_{t_1}) = +1$, then it cannot be that $\val(C^k_{t,c}) \neq -1$. Indeed, since the interior of $C^k_{t,\ell} \cup C^k_{t,c} \cup C^k_{t,r}$ contains a single cut, $\val(C^k_{t,c}) \neq -1$ implies that at least one of $C^k_{t,\ell}$ or $C^k_{t,r}$ is contained in $R^+$. But since $A_{t_1}$ is also contained in $R^+$, this implies $|v_{\alpha^k_t}(R^+) - v_{\alpha^k_t}(R^-)| \geq 2/3 - 1/3 = 1/3 > \eps$, a contradiction. Thus, it must be that $\val(C^k_{t,c}) = -1$. Similarly, we can show that $\val(A_{t_1}) = -1$ implies $\val(C^k_{t,c}) = +1$.
\end{proof}

\paragraph{\bf Gate agents: NAND.} Let $t \in [m]$ be such that $g_t=(g_{t_1},g_{t_2},\textup{NAND})$. Then, for any $k \in [3N]$, the goal of gate agent $\alpha^k_t$ is to enforce the corresponding gate constraint, namely $g^k_t = \val(C^k_{t,c}) = \textup{NAND}(g^k_{t_1}, g^k_{t_2}) = - (g^k_{t_1} \land g^k_{t_2})$. The density function of the valuation $v_{\alpha^k_t}$ is constructed as follows: it has value $2/5$ in $C^k_{t,r}$, value $1/5$ in $A_{t_1} \cup A_{t_2} \cup C^k_{t,\ell}$, and value $0$ everywhere else. The intervals $A_{t_1}, A_{t_2}$ are defined as above in the description of the NOT-gate. See \cref{fig:NAND-gate} for an illustration of the gate.

\begin{figure}
	\centering
	\includegraphics[scale=1.2]{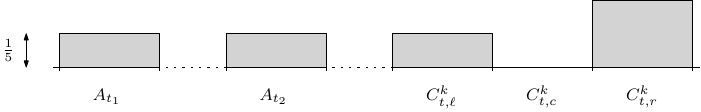}
	\caption{The density function of the valuation of an agent $\alpha_t^k$ implementing a NAND-gate.}\label{fig:NAND-gate}
\end{figure}

\begin{claim}\label{clm:NAND-gate}
For all $t \in [m]$ such that $g_t=(g_{t_1},g_{t_2},\textup{NAND})$, and all $k \in [3N]$, it holds that:
\begin{itemize}
    \item there must be a cut in the interior of $C^k_{t,\ell} \cup C^k_{t,c} \cup C^k_{t,r}$;
    \item if there is exactly one cut in the interior of $C^k_{t,\ell} \cup C^k_{t,c} \cup C^k_{t,r}$, and the inputs to the gate are pure (namely, $g^k_{t_1}, g^k_{t_2} \in \{-1,+1\}$), then the output is pure (i.e., $g^k_t \in \{-1,+1\}$), and
    \begin{itemize}
        \item if the left end of $C^k_t$ lies in $R^+$, then $g^k_t = \textup{NAND}(g^k_{t_1},g^k_{t_2})$;
        \item if the left end of $C^k_t$ lies in $R^-$, then $g^k_t = - \textup{NAND}(-g^k_{t_1},-g^k_{t_2})$.
    \end{itemize}
\end{itemize}
\end{claim}

\begin{proof}
Assume, towards a contradiction, that the interior of $C^k_{t,\ell} \cup C^k_{t,c} \cup C^k_{t,r}$ does not contain any cuts. Then, $C^k_{t,\ell}$ and $C^k_{t,r}$ are both contained in $R^+$ or both contained in $R^-$. This implies that $|v_{\alpha^k_t}(R^+) - v_{\alpha^k_t}(R^-)| \geq 3/5 - 2/5 = 1/5 > \eps$, a contradiction.

Now consider the case where the interior of $C^k_{t,\ell} \cup C^k_{t,c} \cup C^k_{t,r}$ contains exactly one cut. Recall that $g^k_{t_1} = \val(A_{t_1})$, $g^k_{t_2} = \val(A_{t_2})$, and $g^k_t = \val(C^k_{t,c})$. First, assume that the left end of $C^k_t$ lies in $R^+$. If $\val(A_{t_1}) = \val(A_{t_2}) = +1$, then the cut must lie in $C^k_{t,\ell}$. Otherwise, $C^k_{t,\ell}$ lies in $R^+$, just like $A_{t_1}$ and $A_{t_2}$, which implies $|v_{\alpha^k_t}(R^+) - v_{\alpha^k_t}(R^-)| \geq 3/5 - 2/5 = 1/5 > \eps$, a contradiction. Since the cut lies in $C^k_{t,\ell}$, it follows that $C^k_{t,c}$ lies in $R^-$, i.e., $\val(C^k_{t,c}) = -1$, as desired.

If $\val(A_{t_1}) = -1$, then the cut must lie in $C^k_{t,r}$. Indeed, otherwise, $C^k_{t,r}$ lies in $R^-$, just like $A_{t_1}$, which implies $|v_{\alpha^k_t}(R^+) - v_{\alpha^k_t}(R^-)| \geq 3/5 - 2/5 = 1/5 > \eps$, a contradiction. Since the cut lies in $C^k_{t,r}$, it follows that $C^k_{t,c}$ lies in $R^+$, i.e., $\val(C^k_{t,c}) = +1$, as desired. The exact same analysis also applies to the case where $\val(A_{t_2}) = -1$ instead. Thus, we obtain $\val(C^k_{t,c}) = \textup{NAND}(\val(A_{t_1}),\val(A_{t_2}))$.

It remains to consider the setting where the left end of $C^k_t$ lies in $R^-$, instead of $R^+$. The same type of case analysis applied to this setting yields $\val(C^k_{t,c}) = - \textup{NAND}(-\val(A_{t_1}),-\val(A_{t_2}))$.
\end{proof}

\begin{remark}
Note that the proof of \cref{clm:NAND-gate} crucially made use of the fact that $\eps < 1/5$. In fact, it turns out that this is the only point in the reduction where this is needed. The rest of the reduction can be made to work for any $\eps < 1/3$. In particular, it is not hard to see that the proof of \cref{clm:NOT-gate} only made use of the assumption $\eps < 1/3$.
\end{remark}

\paragraph{\bf Feedback agents.} For $i \in [N]$, feedback agent $\gamma_i$ has the following valuation function $v_{\gamma_i}$: the density function of $v_{\gamma_i}$ has value $1/3N$ over $\cup_{k=1}^{3N} C^k_{m-N+i,c}$, and value $0$ everywhere else. Recall that the interval $C^k_{m-N+i,c}$ corresponds to the gate $g_{m-N+i}$ of $\mlambda$, which is the $i$th output of $\mlambda$. For every $k \in [3N]$, define $y^k \in [-1,1]^{N}$ by letting
\[y^k_i := g^k_{m-N+i} = \val(C^k_{m-N+i,c})\]
for all $i \in [N]$. Intuitively, $y^k$ corresponds to the output of the $k$th circuit region $C^k$. By construction of $\gamma_i$, we immediately obtain:

\begin{claim}\label{clm:feedback-agents}
For all $i \in [N]$, it holds that
\[\left| \frac{1}{3N} \sum_{k=1}^{3N} y^k_i \right| \leq \eps.\]
\end{claim}

\begin{proof}
We can write
\begin{equation*}
\begin{split}
\left|v_{\gamma_i}(R^+) - v_{\gamma_i}(R^-)\right| &= \left|\sum_{k=1}^{3N} v_{\gamma_i} (C_{m-N+i,c}^k \cap R^+) - v_{\gamma_i} (C_{m-N+i,c}^k \cap R^-)\right|\\
&= \left|\frac{1}{3N} \sum_{k=1}^{3N} \val(C^k_{m-N+i,c})\right| = \left| \frac{1}{3N} \sum_{k=1}^{3N} y^k_i \right|
\end{split}
\end{equation*}
and at any solution we must have $|v_{\gamma_i}(R^+) - v_{\gamma_i}(R^-)| \leq \eps$.
\end{proof}

We have now completed the construction of the instance \CHinstance. It is easy to check that this construction can be performed in time polynomial in $\sz(\lambda)$.

\subsection{Correctness of the Reduction}

It remains to prove the correctness of the reduction, namely, that from any solution $S = (R^+,R^-)$ to \CHinstance we can extract a solution to the \ndritucker instance $\lambda$. We show this by presenting and proving a sequence of claims.

\begin{claim}\label{clm:C_k-at-least-2m-cuts}
For every $k \in [3N]$, the interior of $C^k$ contains at least $2m$ cuts.
\end{claim}

\begin{proof}
This immediately follows from \cref{clm:cut-in-auxiliary}, \cref{clm:NOT-gate} and \cref{clm:NAND-gate}, by observing that every gate agent and auxiliary agent forces a cut to lie in the interior of some interval, and all these intervals are pairwise disjoint.
\end{proof}

Intuitively, a circuit region $C^k$ will correctly perform computations as long as it does not contain more than $2m$ cuts (and thus, by \cref{clm:C_k-at-least-2m-cuts} above, exactly $2m$ cuts). Furthermore, for the computations to be meaningful, the inputs to the circuit, namely $x^k_{i,j} = \val(I^k_{i,j})$, should also be pure. This motivates defining the ``good'' copies of the circuit as
\[ \small G := \left\{k \in [3N] : \text{the interior of $\overline{C}^k := C^k \cup \left( \bigcup_{(i,j) \in [N] \times [7N]} I^k_{i,j} \right)$ contains at most $2m$ cuts}\right\}.\]

\begin{claim}\label{clm:G-at-least-2N}
It holds that $|G| \geq 2N$.
\end{claim}

\begin{proof}
Note, first of all, that for $k_1 \neq k_2$, the interior of $\overline{C}^{k_1}$ is disjoint from the interior of $\overline{C}^{k_2}$. Furthermore, by \cref{clm:C_k-at-least-2m-cuts} we know that, for each $k \in [3N]$, the interior of $C_k$ contains at least $2m$ cuts. Since there are $n = 3N \cdot 2m + N$ agents, and thus also at most that many cuts, it follows that there remain at most $N$ ``free'' cuts. As a result, the number of $\overline{C}^k$ that contain more than $2m$ cuts can be at most $N$.
\end{proof}

\begin{claim}\label{clm:G-correct-circuit}
For all $k \in G$, we have that $x^k \in \{-1,+1\}^{N \times 7N}$, and
\begin{itemize}
    \item if the left end of $C^k$ lies in $R^+$, then $y^k = \mlambda(x^k)$;
    \item if the left end of $C^k$ lies in $R^-$, then $y^k = - \mlambda(-x^k)$.
\end{itemize}
\end{claim}

\begin{proof}
Since $k \in G$, by definition of $G$ and by \cref{clm:C_k-at-least-2m-cuts}, no cut lies in the interior of $I^k_{i,j}$ for all $(i,j) \in [N] \times [7N]$. As a result, $x^k_{i,j} = \val(I^k_{i,j}) \in \{-1,+1\}$, i.e., $x^k$ is pure.

Consider the case where the left end of $C^k$ lies in $R^+$. By definition of $G$ and by \cref{clm:C_k-at-least-2m-cuts} it follows that the interior of $C^k$ contains exactly $2m$ cuts, and for each $t \in [m]$, the interior of $C^k_t$ contains exactly two cuts, namely one in $C^k_{t,\ell} \cup C^k_{t,c} \cup C^k_{t,r}$, and one in $C^k_{t,a}$. As a result, it holds that for each $t \in [m]$, the left end of $C^k_t$ lies in $R^+$. Thus, since $x^k$ is pure, it follows by \cref{clm:NOT-gate} and \cref{clm:NAND-gate} that the output of the $k$th copy of the first gate is pure, i.e., $g^k_1 \in \{-1,+1\}$, and that the value of the gate is computed correctly. By induction, it follows that $g^k_t \in \{-1,+1\}$ for all $t \in [m]$, and that all the gates are computed correctly. In particular, we obtain that $y^k = \mlambda(x^k)$.

Now, consider the case where the left end of $C^k$ lies in $R^-$. By the same argument as above, it follows that for each $t \in [m]$, the left end of $C^k_t$ lies in $R^-$. As above, since $x^k$ is pure, and by \cref{clm:NOT-gate} and \cref{clm:NAND-gate}, we obtain that the $k$th copy of the first gate $g_1 = (g_{t_1}, g_{t_2}, T)$ is pure, i.e., $g^k_1 \in \{-1,+1\}$, and
\begin{itemize}
    \item if $T = \textup{NOT}$: $g^k_1 = \textup{NOT}(g^k_{t_1}) = - \textup{NOT}(-g^k_{t_1})$;
    \item if $T = \textup{NAND}$: $g^k_1 = - \textup{NAND}(-g^k_{t_1},-g^k_{t_1})$.
\end{itemize}
By induction, it follows that $g^k_t \in \{-1,+1\}$ for all $t \in [m]$, and that $g^k_t = -g_t[-x^k]$, i.e., each gate has the opposite value from the one it would have if the input to the circuit was $-x^k$. In particular, we obtain that $y^k = - \mlambda(-x^k)$.
\end{proof}

We are now ready to complete the proof. Putting everything together, we can prove a stronger version of \cref{clm:G-correct-circuit}.

\begin{claim}\label{clm:all-good}
For all $k \in G$, we have that $x^k \in \{-1,+1\}^{N \times 7N}$, and $y^k = \mlambda(x^k)$.
\end{claim}

\begin{proof}
In order to prove the claim, we consider two distinct cases. First, let us assume that the interior of $I$ contains at least $N$ cuts. Recall that the number of agents is $n = 3N \cdot 2m + N$, and thus the total number of cuts is at most $3N \cdot 2m + N$. Since the interior of $I$ contains at least $N$ cuts, and, for each $k \in [3N]$, the interior of $C^k$ contains at least $2m$ cuts (\cref{clm:C_k-at-least-2m-cuts}), it follows that the interior of $I$ contains exactly $N$ cuts, and, for each $k \in [3N]$, the interior of $C^k$ contains exactly $2m$ cuts. As a result, for each $k \in [3N]$, the left end of $C^k$ lies in $R^+$, because the number of cuts in $C^k$ is even (using the fact that without loss of generality the right end of $R$ lies in $R^+$). By \cref{clm:G-correct-circuit}, it follows that for each $k \in G$, $x^k \in \{-1,+1\}^{N \times 7N}$ and $y^k = \mlambda(x^k)$.

Now, consider the second case, namely that the interior of $I$ contains at most $N-1$ cuts. By \cref{clm:G-correct-circuit} we know that for all $k \in G$, $x^k \in \{-1,+1\}^{N \times 7N}$ and $y^k \in \{\mlambda(x^k), -\mlambda(-x^k)\}$. However, since the interior of $I$ contains at most $N-1$ cuts, it follows by \cref{clm:at-most-N-1-cuts-in-I} that $\mlambda(x^k) = -\mlambda(-x^k)$. Thus, for all $k \in G$, it holds that $y^k = \mlambda(x^k)$.
\end{proof}

\begin{claim}\label{clm:feedback-solution}
The set of points $\{\phi(x^k) : k \in G\}$ yields a solution to the \ndritucker instance $\lambda$.
\end{claim}

\begin{proof}
By \cref{clm:at-most-N-cuts-in-I} and \cref{clm:all-good}, we know that for all $k \in G$, $x^k \in \{-1,+1\}^{N \times 7N}$ and $y^k = \mlambda(x^k) = \lambda(\phi(x^k))$. Furthermore, for all $k_1, k_2 \in G$, we have $\|\phi(x^{k_1}) - \phi(x^{k_2})\|_\infty \leq 1$. Thus, it remains to show that the points in $\{x^k : k \in G\}$ cover all the labels of $\mlambda$. Towards a contradiction, assume that this is not the case. Then, there exists $i \in [N]$ and $b \in \{-1,+1\}$ such that $y^k_i = b$ for all $k \in G$. But then, since $|G| \geq 2N$ (\cref{clm:G-at-least-2N}), and $|y^k_i| \leq 1$ for all $k \in [3N]$,
\[\left| \frac{1}{3N} \sum_{k=1}^{3N} y^k_i \right| \geq \frac{1}{3N} (|G| - (3N - |G|)) \geq 1/3 > \eps\]
which contradicts \cref{clm:feedback-agents}, namely, the feedback agent $\gamma_i$ cannot be satisfied in that case. It follows that the points in $\{x^k : k \in G\}$ do indeed cover all the labels of $\mlambda$. As a result, we can extract a solution to $\lambda$ from $\{\phi(x^k) : k \in G\}$ by using \cref{lem:r-to-N}.
\end{proof}

Finally, note that, given a solution $S$ of \CHinstance, we can in polynomial time compute $G$, then $\{\phi(x^k) : k \in G\}$, and finally use \cref{lem:r-to-N} to extract $N$ points that are a solution to $\lambda$. This completes the proof of correctness for the reduction.

\subsection{Extension to 3-Block Uniform Valuations}\label{sec:3-block}

The proof that we presented above can be modified to prove \cref{thm:main-result-3-block}, namely that the result holds even if we restrict the valuations to be 3-block uniform. Recall that an agent has a 3-block uniform valuation function if the density function of the valuation is non-zero in at most three intervals, and in each such interval it has the same non-zero value. \citet{FRHSZ20-consensus-easier} have proved that the problem remains \ppa/-complete even for 2-block uniform valuations, but their hardness result only holds for polynomially small $\eps$.

The following modifications to the proof of \cref{thm:main-result} are needed to obtain \cref{thm:main-result-3-block}:

\begin{itemize}
    \item \textbf{Number of copies:} Instead of $3N$ copies of the circuit, we use $20N$ copies of the circuit. In particular, every interval $I_{i,j}$ is now subdivided into intervals $I_{i,j}^1, \dots, I_{i,j}^{20N}$.
    
    \item \textbf{NOT-gates:} Auxiliary agents and NOT-gate agents already have 3-block uniform valuations. Thus, no change is needed there.
    
    \item \textbf{NAND-gates:} Unfortunately, NAND-gate agents are not 3-block uniform. To address this, we modify each NAND-gate agent as follows: the density function of the valuation has value $1/5$ in each of the three intervals $A_{t_1}$, $A_{t_2}$, and $C_{t,\ell}^k \cup C_{t,c}^k \cup C_{t,r}^k$, and value $0$ everywhere else. This valuation function is now 3-block uniform, but the value of the output of the gate is no longer encoded in the standard way. Indeed, letting $J_t$ denote the interval of length $\delta := (1/5-\eps)/2 > 0$ centered around the point $C_{t,\ell}^k \cap C_{t,c}^k$, we can prove an analogue of \cref{clm:NAND-gate}, where the output value is no longer stored in $C_{t,c}^k$, but in $J_t$ instead (i.e., $J_t \subseteq R^+$ means that the output is $+1$, and $J_t \subseteq R^-$ means that the output is $-1$). As a result, some additional modifications are needed to correctly read the output of such a gate.
    
    Without loss of generality, we can assume that in the circuit $\mlambda$, every NAND-gate is always followed by a NOT-gate, i.e., the output of a NAND-gate can only be used as an input to a NOT-gate (and, in particular, can also not be an output of the circuit). This can easily be ensured by introducing two consecutive NOT-gates wherever that is needed. With this in hand, for each NOT-gate $g_{t'}$ that takes as input the output of a NAND-gate $g_{t}$, we will modify the corresponding NOT-gate agent as follows: the density function of the valuation has a block of length $\delta$ and height $1/3\delta$ in $J_t$, and two more such blocks, one in each of $C_{t',\ell}^k$ and $C_{t',r}^k$ (anywhere inside those intervals). We can prove an analogue of \cref{clm:NOT-gate} to show that this agent behaves as a standard NOT-gate, except that the input is read from $J_t$ instead of $C_{t,c}^k$. See \cref{fig:3-block-NAND-gate} for an illustration of the modified NAND-gate, together with its corresponding modified NOT-gate.
    
    \item \textbf{Feedback agents:} The last remaining agents that do not have 3-block uniform valuations are the feedback agents $\gamma_1, \dots, \gamma_N$. To address this we modify the feedback mechanism as follows. We add an ``Output region'' $O$ between the input region $I$ and the circuit region $C$. Interval $O$ is subdivided into intervals $O_1, \dots, O_n$, and each $O_i$ is subdivided into intervals $O_i^1, \dots, O_i^{20N}$. Finally, each $O_i^k$ is subdivided into three intervals of length $1$: $O_{i,\ell}^k$, $O_{i,c}^k$, and $O_{i,r}^k$. For each $i \in [N]$ and $k \in [20N]$, using an additional NOT-gate based in $O_i^k$, we copy the value of the $i$th output of the $k$th copy of the circuit into the interval $O_{i,c}^k$. Finally, we define the feedback agent $\gamma_i$ as follows: the density function is uniform over $O_i$, i.e., it has value $1/60N$ over $O_i$, and value $0$ elsewhere.
    
    With this modified construction we can then prove an analogue of \cref{clm:feedback-solution}. The main observation is that if, say, $y_i^k = +1$ for all $k \in G$, then $\mu(O_i \cap R^+) \geq 2|G|$ and $\mu(O_i \cap R^-) \leq |G| + 3(20N-|G|)$. Using the fact that $|G| \geq 20N - N$, it follows that agent $\gamma_i$ is not satisfied, since
    \[\frac{1}{60N} (\mu(O_i \cap R^+) - \mu(O_i \cap R^-)) \geq \frac{1}{60N} (4|G| - 60N) \geq \frac{80N - 4N - 60N}{60N} = \frac{4}{15} > \eps.\]
    Note that here we crucially used the fact that we now have $20N$ copies instead of just $3N$.
\end{itemize}

\begin{figure}
	\centering
	\includegraphics[scale=0.915]{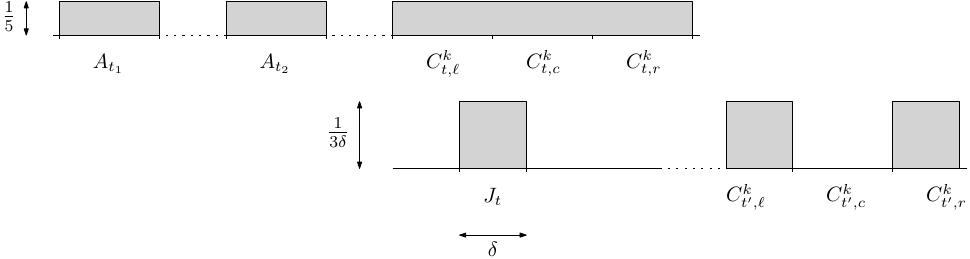}
	\caption{The valuation function of an agent implementing a modified NAND-gate (top), and the valuation function of an agent implementing a modified NOT-gate reading the output of the NAND-gate (bottom).}\label{fig:3-block-NAND-gate}
\end{figure}

\section{Conclusion}
So far, \eps-\conhalv has been the starting point for every \ppa/-hardness result for problems that do not include a circuit in their definition. We have resolved the complexity of the problem for constant approximations by showing that it is \ppa/-complete for every constant $\eps < 1/5$. We expect that this will be very useful to obtain further strong inapproximability results for \ppa/ problems.

There are several remaining questions related to \conhalv.

\begin{itemize}
    \item \textbf{Improve $\boldsymbol{\eps}$ beyond $\boldsymbol{1/5}$.} The current bottleneck of our technique is the NAND gate. We conjecture that the \eps we derive for a NAND gate implemented via a single agent, or even a constant number of agents, is optimal. Hence, we believe that a new technique would be needed to get \ppa/-hardness for a larger \eps.
    
    The NAND gate, and, in fact, any gadget taking two bits as input and having a non-trivial output bit (e.g., not just copying the first input bit), has to balance out two constraints. Let $I$ denote the subinterval of the consensus-halving interval that encodes the input(s) of the gate, and $O$ the subinterval that encodes the output. Then the two constraints are the following.
    \begin{enumerate}
        \item The agent encoding the gate has to have enough value in $I$ (with respect to $\eps$), so that ``what happens in I has some effect on the agent''.
        \item The agent encoding the gate has to have enough value in $O$ (with respect to $\eps$), so that there is necessarily a cut in $O$. This is to avoid extra stray cuts, which the current reduction framework (like all previous ones) cannot handle.
    \end{enumerate}
    Together with the fact that the agent's valuation has to be normalized to 1, these two constraints yield that $\eps < 1/5$ for a Boolean gate with two inputs, and $\eps < 1/3$ for a Boolean gate with one input.
    
    \item \textbf{Prove an upper bound.} So far, no algorithm is known for solving \eps-\conhalv (with $n$ cuts) for some constant $\eps < 1$, even for piecewise constant valuations with positive value on two intervals only. The only known upper bounds are for very special cases~\citep{DeligkasFMS22,FRHSZ20-consensus-easier} or with additional cuts \citep{AlonG21}.
    
    \item {\bf Constant number of agents.} For a constant number of agents with explicitly represented piecewise constant valuations (as modelled in this paper), \eps-\conhalv can be solved in polynomial time by a simple enumeration algorithm \citep{FRHSZ20-consensus-easier}. But what is the complexity of the problem if we are not given the whole valuations upfront, but instead can only efficiently evaluate them? In~\citep{DeligkasFH22-consensus-constant} it was proven that in this setting the problem is \ppa/-complete for 3 agents with {\em non-additive} valuations. It seems that proving hardness for the more standard additive valuation model would require radically new ideas.
    
    \item {\bf Necklaces with few beads per colour.} What is the computational complexity of \necklace with a constant number of beads per colour? Our hardness result, which directly uses the reduction presented by \citet{FRG18-Consensus}, constructs necklaces with polynomially-many beads for each colour. There appears to be no straightforward way to reduce this number to a constant. On the other hand, to the best of our knowledge, there is no efficient algorithm that solves \necklace when every colour appears at most four times.
\end{itemize}

\subsubsection*{Acknowledgements}

The first author was supported by EPSRC grant EP/X039862/1 ``NAfANE: New Approaches for Approximate Nash Equilibria''. The second author was supported by EPSRC grant EP/W014750/1 ``New Techniques for Resolving Boundary Problems in Total Search''. The work was completed while the fourth author was affiliated with the group of Operations Research at the Technical University of Munich. This author's work was supported by the Alexander von Humboldt Foundation with funds from the German Federal Ministry of Education and Research (BMBF).

\bibliographystyle{plainnat}
\bibliography{references}

\begin{thebibliography}{38}
\providecommand{\natexlab}[1]{#1}
\providecommand{\url}[1]{\texttt{#1}}
\expandafter\ifx\csname urlstyle\endcsname\relax
  \providecommand{\doi}[1]{doi: #1}\else
  \providecommand{\doi}{doi: \begingroup \urlstyle{rm}\Url}\fi

\bibitem[Aharoni et~al.(2017)Aharoni, Alon, Berger, Chudnovsky, Kotlar, Loebl,
  and Ziv]{aharoni2017fair}
Ron Aharoni, Noga Alon, Eli Berger, Maria Chudnovsky, Dani Kotlar, Martin
  Loebl, and Ran Ziv.
\newblock Fair representation by independent sets.
\newblock In \emph{A Journey Through Discrete Mathematics}, pages 31--58.
  Springer, 2017.
\newblock \doi{10.1007/978-3-319-44479-6_2}.

\bibitem[Aisenberg et~al.(2020)Aisenberg, Bonet, and
  Buss]{AisenbergBB20-2D-Tucker}
James Aisenberg, Maria~Luisa Bonet, and Sam Buss.
\newblock 2-{D} {T}ucker is {PPA} complete.
\newblock \emph{Journal of Computer and System Sciences}, 108:\penalty0
  92--103, 2020.
\newblock \doi{10.1016/j.jcss.2019.09.002}.

\bibitem[Alishahi and Meunier(2017)]{alishahi2017fair}
Meysam Alishahi and Fr{\'e}d{\'e}ric Meunier.
\newblock Fair splitting of colored paths.
\newblock \emph{The Electronic Journal of Combinatorics}, 24\penalty0
  (3):\penalty0 P3.41, 2017.
\newblock \doi{10.37236/7079}.

\bibitem[Alon(1987)]{Alon87-necklace}
Noga Alon.
\newblock Splitting necklaces.
\newblock \emph{Advances in Mathematics}, 63\penalty0 (3):\penalty0 247--253,
  1987.
\newblock \doi{10.1016/0001-8708(87)90055-7}.

\bibitem[Alon and Graur(2021)]{AlonG21}
Noga Alon and Andrei Graur.
\newblock Efficient splitting of necklaces.
\newblock In \emph{Proceedings of the 48th International Colloquium on
  Automata, Languages, and Programming (ICALP)}, pages 14:1--14:17, 2021.
\newblock \doi{10.4230/LIPIcs.ICALP.2021.14}.

\bibitem[Alon and West(1986)]{alon1986borsuk}
Noga Alon and Douglas~B. West.
\newblock {The Borsuk-Ulam Theorem and Bisection of Necklaces}.
\newblock \emph{Proceedings of the American Mathematical Society}, 98\penalty0
  (4):\penalty0 623--628, 1986.
\newblock \doi{10.2307/2045739}.

\bibitem[Barba et~al.(2019)Barba, Pilz, and Schnider]{BPS19}
Luis Barba, Alexander Pilz, and Patrick Schnider.
\newblock Sharing a pizza: bisecting masses with two cuts.
\newblock \emph{arXiv preprint}, abs/1904.02502, 2019.
\newblock URL \url{https://arxiv.org/abs/1904.02502}.

\bibitem[Batziou et~al.(2021)Batziou, Hansen, and
  H{\o}gh]{BatziouHH21-consensus-BBU}
Eleni Batziou, Kristoffer~Arnsfelt Hansen, and Kasper H{\o}gh.
\newblock Strong approximate {C}onsensus {H}alving and the {B}orsuk-{U}lam
  theorem.
\newblock In \emph{Proceedings of the 48th International Colloquium on
  Automata, Languages, and Programming (ICALP)}, pages 24:1--24:20, 2021.
\newblock \doi{10.4230/LIPIcs.ICALP.2021.24}.

\bibitem[Black et~al.(2020)Black, Cetin, Frick, Pacun, and
  Setiabrata]{black2020fair}
Alexander Black, Umur Cetin, Florian Frick, Alexander Pacun, and Linus
  Setiabrata.
\newblock Fair splittings by independent sets in sparse graphs.
\newblock \emph{Israel Journal of Mathematics}, 236:\penalty0 603--627, 2020.
\newblock \doi{10.1007/s11856-020-1980-5}.

\bibitem[Chen et~al.(2009)Chen, Deng, and Teng]{ChenDT09-Nash}
Xi~Chen, Xiaotie Deng, and Shang-Hua Teng.
\newblock Settling the complexity of computing two-player {N}ash equilibria.
\newblock \emph{Journal of the ACM}, 56\penalty0 (3):\penalty0 14:1--14:57,
  2009.
\newblock \doi{10.1145/1516512.1516516}.

\bibitem[Daskalakis et~al.(2009)Daskalakis, Goldberg, and
  Papadimitriou]{DaskalakisGP09-Nash}
Constantinos Daskalakis, Paul~W. Goldberg, and Christos~H. Papadimitriou.
\newblock The complexity of computing a {N}ash equilibrium.
\newblock \emph{SIAM Journal on Computing}, 39\penalty0 (1):\penalty0 195--259,
  2009.
\newblock \doi{10.1137/070699652}.

\bibitem[Deligkas et~al.(2021)Deligkas, Fearnley, Melissourgos, and
  Spirakis]{deligkas2021BU}
Argyrios Deligkas, John Fearnley, Themistoklis Melissourgos, and Paul~G.
  Spirakis.
\newblock Computing exact solutions of consensus halving and the
  {Borsuk}-{Ulam} theorem.
\newblock \emph{Journal of Computer and System Sciences}, 117:\penalty0 75--98,
  2021.
\newblock \doi{10.1016/j.jcss.2020.10.006}.

\bibitem[Deligkas et~al.(2022{\natexlab{a}})Deligkas, Fearnley, and
  Melissourgos]{DeligkasFM22-pizza}
Argyrios Deligkas, John Fearnley, and Themistoklis Melissourgos.
\newblock Pizza sharing is {PPA}-hard.
\newblock In \emph{Proceedings of the 36th AAAI Conference on Artificial
  Intelligence (AAAI)}, pages 4957--4965, 2022{\natexlab{a}}.
\newblock \doi{10.1609/aaai.v36i5.20426}.

\bibitem[Deligkas et~al.(2022{\natexlab{b}})Deligkas, Fearnley, Melissourgos,
  and Spirakis]{DeligkasFMS22}
Argyrios Deligkas, John Fearnley, Themistoklis Melissourgos, and Paul~G.
  Spirakis.
\newblock Approximating the existential theory of the reals.
\newblock \emph{Journal of Computer and System Sciences}, 125:\penalty0
  106--128, 2022{\natexlab{b}}.
\newblock \doi{10.1016/j.jcss.2021.11.002}.

\bibitem[Deligkas et~al.(2022{\natexlab{c}})Deligkas, Filos-Ratsikas, and
  Hollender]{DeligkasFH22-consensus-constant}
Argyrios Deligkas, Aris Filos-Ratsikas, and Alexandros Hollender.
\newblock Two's company, three's a crowd: Consensus-halving for a constant
  number of agents.
\newblock \emph{Artificial Intelligence}, 313, 2022{\natexlab{c}}.
\newblock \doi{10.1016/j.artint.2022.103784}.
\newblock 103784.

\bibitem[Deligkas et~al.(2024)Deligkas, Fearnley, Hollender, and
  Melissourgos]{DeligkasFHM24-pure-circuit}
Argyrios Deligkas, John Fearnley, Alexandros Hollender, and Themistoklis
  Melissourgos.
\newblock Pure-circuit: Tight inapproximability for {PPAD}.
\newblock \emph{Journal of the ACM}, 71\penalty0 (5):\penalty0 31:1--31:48,
  2024.
\newblock \doi{10.1145/3678166}.

\bibitem[Deng et~al.(2017)Deng, Feng, and Kulkarni]{dengoctahedral}
Xiaotie Deng, Zhe Feng, and Rucha Kulkarni.
\newblock Octahedral {T}ucker is {PPA}-complete.
\newblock Technical Report TR17-118, Electronic Colloquium on Computational
  Complexity (ECCC), 2017.
\newblock URL \url{https://eccc.weizmann.ac.il/report/2017/118/}.

\bibitem[Etessami and Yannakakis(2010)]{EtessamiY10-FIXP}
Kousha Etessami and Mihalis Yannakakis.
\newblock On the complexity of {N}ash equilibria and other fixed points.
\newblock \emph{SIAM Journal on Computing}, 39\penalty0 (6):\penalty0
  2531--2597, 2010.
\newblock \doi{10.1137/080720826}.

\bibitem[Filos-Ratsikas and Goldberg(2018)]{FRG18-Consensus}
Aris Filos-Ratsikas and Paul~W. Goldberg.
\newblock Consensus halving is {PPA}-complete.
\newblock In \emph{Proceedings of the 50th ACM Symposium on Theory of Computing
  (STOC)}, pages 51--64, 2018.
\newblock \doi{10.1145/3188745.3188880}.

\bibitem[Filos-Ratsikas and Goldberg(2019)]{FRG19-Necklace}
Aris Filos-Ratsikas and Paul~W. Goldberg.
\newblock The complexity of splitting necklaces and bisecting ham sandwiches.
\newblock In \emph{Proceedings of the 51st ACM Symposium on Theory of Computing
  (STOC)}, pages 638--649, 2019.
\newblock \doi{10.1145/3313276.3316334}.

\bibitem[Filos-Ratsikas and Goldberg(2023)]{FRG22-NS-CH-ham}
Aris Filos-Ratsikas and Paul~W. Goldberg.
\newblock The complexity of {Necklace Splitting}, {Consensus-Halving}, and
  {Discrete Ham Sandwich}.
\newblock \emph{SIAM Journal on Computing}, 52\penalty0 (2):\penalty0
  STOC19--200--STOC19--268, 2023.
\newblock \doi{10.1137/20M1312678}.

\bibitem[Filos-Ratsikas et~al.(2018)Filos-Ratsikas, Frederiksen, Goldberg, and
  Zhang]{FRFGZ18-consensus-hardness}
Aris Filos-Ratsikas, S{\o}ren Kristoffer~Still Frederiksen, Paul~W. Goldberg,
  and Jie Zhang.
\newblock Hardness results for {C}onsensus-{H}alving.
\newblock In \emph{Proceedings of the 43rd International Symposium on
  Mathematical Foundations of Computer Science (MFCS)}, pages 24:1--24:16,
  2018.
\newblock \doi{10.4230/LIPIcs.MFCS.2018.24}.

\bibitem[Filos-Ratsikas et~al.(2021)Filos-Ratsikas, Hollender, Sotiraki, and
  Zampetakis]{FRHSZ21-necklace}
Aris Filos-Ratsikas, Alexandros Hollender, Katerina Sotiraki, and Manolis
  Zampetakis.
\newblock A topological characterization of modulo-$p$ arguments and
  implications for necklace splitting.
\newblock In \emph{Proceedings of the 32nd ACM-SIAM Symposium on Discrete
  Algorithms (SODA)}, pages 2615--2634, 2021.
\newblock \doi{10.1137/1.9781611976465.155}.

\bibitem[Filos-Ratsikas et~al.(2023)Filos-Ratsikas, Hollender, Sotiraki, and
  Zampetakis]{FRHSZ20-consensus-easier}
Aris Filos-Ratsikas, Alexandros Hollender, Katerina Sotiraki, and Manolis
  Zampetakis.
\newblock {Consensus-Halving}: Does it ever get easier?
\newblock \emph{SIAM Journal on Computing}, 52\penalty0 (2):\penalty0 412--451,
  2023.
\newblock \doi{10.1137/20m1387493}.

\bibitem[Goldberg and West(1985)]{goldberg1985bisection}
Charles~H. Goldberg and Douglas~B. West.
\newblock {Bisection of Circle Colorings}.
\newblock \emph{SIAM Journal on Algebraic Discrete Methods}, 6\penalty0
  (1):\penalty0 93--106, 1985.
\newblock \doi{10.1137/0606010}.

\bibitem[Goldberg et~al.(2022)Goldberg, Hollender, Igarashi, Manurangsi, and
  Suksompong]{GoldbergHIMS22-consensus-items}
Paul~W. Goldberg, Alexandros Hollender, Ayumi Igarashi, Pasin Manurangsi, and
  Warut Suksompong.
\newblock Consensus halving for sets of items.
\newblock \emph{Mathematics of Operations Research}, 47\penalty0 (4):\penalty0
  3357--3379, 2022.
\newblock \doi{10.1287/moor.2021.1249}.

\bibitem[Haviv(2021)]{haviv2020complexity}
Ishay Haviv.
\newblock The complexity of finding fair independent sets in cycles.
\newblock In \emph{Proceedings of the 12th Innovations in Theoretical Computer
  Science Conference (ITCS)}, pages 4:1--4:14, 2021.
\newblock \doi{10.4230/LIPIcs.ITCS.2021.4}.

\bibitem[Haviv(2022)]{Haviv2022-fair-cycles}
Ishay Haviv.
\newblock The complexity of finding fair independent sets in cycles.
\newblock \emph{computational complexity}, 31\penalty0 (2), 2022.
\newblock \doi{10.1007/s00037-022-00233-6}.

\bibitem[Hobby and Rice(1965)]{hobby1965moment}
Charles~R. Hobby and John~R. Rice.
\newblock {A moment problem in L1 approximation}.
\newblock \emph{Proceedings of the American Mathematical Society}, 16\penalty0
  (4):\penalty0 665--670, 1965.
\newblock \doi{10.2307/2033900}.

\bibitem[Hollender(2021)]{Hollender21-ppa-k}
Alexandros Hollender.
\newblock The classes {PPA}-$k$: Existence from arguments modulo $k$.
\newblock \emph{Theoretical Computer Science}, 885:\penalty0 15--29, 2021.
\newblock \doi{10.1016/j.tcs.2021.06.016}.

\bibitem[Hubard and Karasev(2020)]{HK20}
Alfredo Hubard and Roman Karasev.
\newblock Bisecting measures with hyperplane arrangements.
\newblock \emph{Mathematical Proceedings of the Cambridge Philosophical
  Society}, 169\penalty0 (3):\penalty0 639--647, 2020.
\newblock \doi{10.1017/S0305004119000380}.

\bibitem[Karasev et~al.(2016)Karasev, Rold{\'a}n-Pensado, and
  Sober{\'o}n]{KRS16}
Roman~N. Karasev, Edgardo Rold{\'a}n-Pensado, and Pablo Sober{\'o}n.
\newblock Measure partitions using hyperplanes with fixed directions.
\newblock \emph{Israel Journal of Mathematics}, 212:\penalty0 705--728, 2016.
\newblock \doi{10.1007/s11856-016-1303-z}.

\bibitem[Megiddo and Papadimitriou(1991)]{MegiddoP91-TFNP}
Nimrod Megiddo and Christos~H. Papadimitriou.
\newblock On total functions, existence theorems and computational complexity.
\newblock \emph{Theoretical Computer Science}, 81\penalty0 (2):\penalty0
  317--324, 1991.
\newblock \doi{10.1016/0304-3975(91)90200-L}.

\bibitem[Papadimitriou(1994)]{Papadimitriou94-TFNP-subclasses}
Christos~H. Papadimitriou.
\newblock On the complexity of the parity argument and other inefficient proofs
  of existence.
\newblock \emph{Journal of Computer and System Sciences}, 48\penalty0
  (3):\penalty0 498--532, 1994.
\newblock \doi{10.1016/S0022-0000(05)80063-7}.

\bibitem[Rubinstein(2018)]{Rubinstein18-Nash-inapproximability}
Aviad Rubinstein.
\newblock Inapproximability of {N}ash equilibrium.
\newblock \emph{SIAM Journal on Computing}, 47\penalty0 (3):\penalty0 917--959,
  2018.
\newblock \doi{10.1137/15M1039274}.

\bibitem[Schnider(2022)]{schnider2021complexity}
Patrick Schnider.
\newblock The complexity of sharing a pizza.
\newblock \emph{Computing in Geometry and Topology}, 1\penalty0 (1):\penalty0
  4:1--4:19, 2022.
\newblock \doi{10.57717/CGT.V1I1.5}.

\bibitem[Simmons and Su(2003)]{SS03-Consensus}
Forest~W. Simmons and Francis~E. Su.
\newblock Consensus-halving via theorems of {B}orsuk-{U}lam and {T}ucker.
\newblock \emph{Mathematical social sciences}, 45\penalty0 (1):\penalty0
  15--25, 2003.
\newblock \doi{10.1016/S0165-4896(02)00087-2}.

\bibitem[Tucker(1945)]{tucker1945some}
Albert~W. Tucker.
\newblock Some topological properties of disk and sphere.
\newblock In \emph{Proceedings of the First Canadian Math. Congress, Montreal},
  pages 286--309. University of Toronto Press, 1945.

\end{thebibliography}

\appendix

\section{Additional Cuts}
\label{sec:app:extra-cuts}

In this section, we prove \cref{thm:extra_cuts}, which we restate here for convenience.

\extracuts*

\begin{proof}
We use the same approach used to prove Corollary~3.2 in \citep{FRHSZ20-consensus-easier}. Let $\eps < 1/5$ and $\delta \in (0,1]$. Consider an instance CH of $\eps$-\conhalv where all agents have 3-block uniform valuations and let $n$ denote the number of agents. By \cref{thm:main-result-3-block} we know that finding a solution in such instances using at most $n$ cuts is \ppa/-hard. We show how to reduce this to the problem of finding a solution in an instance with $N$ agents (who all still have 3-block uniform valuations), but using up to $N + N^{1-\delta}$ cuts.

We construct an instance CH$'$ that is composed of $c$ completely disjoint copies of the instance CH. This means that the instance is defined on the interval $[0,c]$, with the $i$th copy living in subinterval $[i-1,i]$. (It is easy to renormalize the instance CH$'$ to be defined on interval $[0,1]$ without affecting any of the arguments.) The number of copies $c$ will be fixed below, but importantly it must be polynomial in the size of the original instance CH. The number of agents in the new instance is $N = c \cdot n$.

Consider any solution of the new instance CH$'$ that uses at most $N + (c-1) = c \cdot n + (c-1)$ cuts. By the pigeonhole principle, at least one of the copies will be cut by at most $n$ cuts. Since the copies are disjoint, this will yield a solution to the original instance CH.

It remains to show that we can pick a number of copies $c$ that is polynomial in the size of the original instance CH, while also ensuring that $N + (c-1) \geq N + N^{1-\delta}$. Letting $c := n^k$ and recalling that $N = c \cdot n$, this inequality can be rewritten as $n^k-1 \geq n^{(k+1)(1-\delta)}$. Now since $\delta > 0$, there exists a sufficiently large $k$ (depending only on $\delta$) such that the inequality is satisfied for sufficiently large $n$.
\end{proof}

\end{document}